\documentclass{article}

\usepackage[linesnumbered,ruled,vlined]{algorithm2e}
\RequirePackage{amssymb}

\usepackage{enumitem}
\usepackage{graphicx}
\usepackage[T1]{fontenc}
\usepackage{notation}

\usepackage{tabularx}
\usepackage{makecell}
\usepackage{longtable}

\newcommand{\cmarkcell}{\textcolor{green!60!black}{\cmark}}
\newcommand{\xmarkcell}{\textcolor{red!70!black}{\xmark}}

\usetikzlibrary{arrows.meta,positioning,calc}

\SetKwInput{KwReturn}{Return}

\usepackage{arxiv}

\title{Equilibrium with Internal Transfers}
\author{Mingyang Liu $^1$, Gabriele Farina$^1$, Asuman Ozdaglar $^1$\\
    $^1$ LIDS, EECS, Massachusetts Institute of Technology\\
    $^1$ \texttt{\{liumy19,gfarina,asuman\}@mit.edu}\\
}

\begin{document}

\maketitle

\begin{abstract}
    Nash equilibrium (NE) arises from selfish utility maximization, yet its social welfare can be arbitrarily far from optimal. Moreover, computing an NE is intractable in general. We study augmented game models in which players use budget-balanced internal transfers to improve incentives before play. We first introduce \emph{Self-Enforcing Transfer Equilibrium} (SETE), where players commit to nonnegative peer-to-peer transfers that are paid only if the recipient does not deviate from a prescribed strategy. For polymatrix games, we show that every stationary point of the social welfare function, in particular any socially optimal strategy profile, can be sustained as a SETE. This induces a Nash equilibrium in the agent normal form of the corresponding augmented game. We further propose a polynomial-time algorithm and a decentralized learning dynamic to compute such product-form equilibria. We then introduce \emph{Mediated Self-Enforcing Transfer Equilibrium} (M-SETE), where a mediator makes both the payment schedule and the prescribed strategies binding offers. This additional enforcement resolves the agent-normal-form limitation: an M-SETE is a Nash equilibrium of the augmented game itself, not merely of its agent normal form, and any socially optimal strategy profile can be supported as an M-SETE in any finite game while preserving budget balance. Thus, internal transfers improve welfare and computation while preserving independent play on the equilibrium path. When full sequential-game stability is required, binding mediation provides the corresponding implementation.\footnote{Accepted for presentation at the ACM Conference on Economics and Computation (EC 26).}
\end{abstract}

\section{Introduction}

Nash equilibrium (NE) \citep{nash1950equilibrium-nash-def} is the canonical solution concept for noncooperative games: it is fully decentralized (no mediator is required), and each player chooses an action independently.
When players cannot coordinate to correlate their behavior, NE is the natural concept.
However, NE can be highly inefficient from a societal perspective.
Even when every player optimizes their own utility, equilibrium outcomes may exhibit low social welfare (the sum of players' utilities), a phenomenon captured by the Price of Anarchy (PoA) \citep{roughgarden2002bad-POA}.

A long line of work seeks to improve welfare by enlarging the set of equilibrium outcomes.
One approach is to allow \emph{correlation} across players' actions by recommending joint actions drawn from a correlated distribution, yielding correlated equilibrium (CE) \citep{aumann1974subjectivity-CE-def} and coarse correlated equilibrium (CCE) \citep{moulin1978strategically-CCE-def}.
Implementing such outcomes typically requires a mediator to sample from the target distribution and privately recommend actions to players.
Moreover, the mediator must be trusted to follow the prescribed distribution, and correlated recommendations inherently give up the independence of players' strategies.

A second approach is to perturb players' utilities so that higher-welfare outcomes become equilibria.
For instance, \citet{sandholm2002evolutionary-pigouvian,sandholm2005negative-pigouvian,monderer2003k-implementation} study external payments such as Pigouvian taxes or subsidies, imposed by a government or mediator, to reshape incentives and steer play toward socially preferred outcomes.
Such external-payment schemes, however, rely on third-party intervention and inject transfers from outside the system, which means they are not budget-balanced.

More recently, \citet{jackson2005endogenous,kolumbus2024paying-do-better} study \emph{internal transfers}, where payments occur strictly between existing players and remain budget-balanced.
Their mechanisms allow each player to specify payment rules contingent on the realized \emph{joint} action, which can effectively redistribute utilities subject only to preserving social welfare.
This generality comes at a cost: it assumes that players know one another's utility functions in order to compute the required payments and introduces decision variables that scale exponentially with the number of players.
In contrast, the payment plan in our proposed \emph{Self-Enforcing Transfer Equilibrium (SETE)} depends only on the \emph{sender} and \emph{receiver} of each transfer, and it avoids requiring players to know others' utility functions.

The welfare losses of NE are only part of the difficulty. Even if we accept NE as the appropriate behavioral benchmark, it may be hard to realize in practice because finding an equilibrium can be computationally intractable. In particular, computing an NE is {\tt PPAD}-complete in general \citep{daskalakis2009complexity-PPAD,chen2006settling-ppad}, suggesting that even if NE is conceptually appealing, it may be impractical to compute or learn in complex environments. These two obstacles, \emph{welfare loss} and \emph{computational intractability}, motivate the search for equilibrium notions that retain NE's decentralization while enabling better outcomes and efficient computation.

In this work, we formalize four desiderata for a practical equilibrium notion:
\begin{enumerate}
    \item No Mediator: the mechanism should be implementable without a trusted third party.
    \item Independent Strategies: players need not rely on correlated recommendations.
    \item Budget-Balance: any transfers must be internal to the system.
    \item Computability: the equilibrium can be found efficiently.
\end{enumerate}
\Cref{table:comparison} summarizes how existing approaches trade off the four desiderata above. SETE is designed to satisfy all four desiderata. In \Cref{section:mediated-SETE}, we study a mediated variant that deliberately relaxes the no mediator desideratum. By making both the payment schedule and the prescribed strategies binding offers, the binding mechanism strengthening the guarantee from an NE in the agent normal form of the augmented game to an NE of the augmented game itself.

\begin{example}
\label{example:PD}
    Consider the Prisoner's Dilemma, as shown in \Cref{table:PD}. In the standard setup, Defect (D) is the dominant strategy. However, if the Row Player chooses to Defect (D), they might be willing to pay the Column Player to play Cooperate (C) rather than Defect, as illustrated in \Cref{fig:PD}.
    \begin{enumerate}
        \item Willingness to Pay: The Row Player prevents a loss of $0.8$ by ensuring the Column Player does not deviate, making them willing to pay up to this amount.
        \item Incentive to Accept: The Row Player actually only needs to pay $0.2$ to make it profitable for the Column Player to stick to Cooperate (C).
    \end{enumerate}
\end{example}

\Cref{example:PD} highlights that side payments can be individually rational even among fully selfish players: if another player's deviation would hurt me, it can be optimal to compensate them for not deviating.
This motivates a natural question:
\begin{quote}
    When would selfish players be willing to pay others?
\end{quote}
The key is that transfers serve as a commitment device rather than a charitable act.
A payer uses payments to internalize the externality of a potential deviation and to make the target behavior incentive compatible. In particular, the payer is willing to pay up to the loss they would otherwise suffer, while the receiver accepts so long as the payment covers their gain from deviating.

We propose the Self-Enforcing Transfer Equilibrium (SETE). Unlike a standard Nash equilibrium, a SETE is defined by a pair $(\pi, \bp)$, consisting of a strategy profile $\pi$ and a payment plan $\bp$. The formal definition is provided in \Cref{def:SETE-def-label}.
At a high level, SETE is governed by two principles:
\begin{enumerate}
    \item Conditional Payments: A player $i$ offers a payment $p_{i\rightarrow j}$ to player $j$ conditional on $j$ \emph{not} deviating from the assigned strategy $\pi_j$. Intuitively, $p_{i\rightarrow j}$ compensates $j$ for staying put and is tied to the externality that $j$'s deviation could impose on $i$.
    \item Incentive Compatibility: The payment $p_{i\rightarrow j}$ is bounded by player $i$'s willingness to pay---namely, the maximum loss $i$ might suffer if $j$ deviates.
    This bound is deliberately pessimistic: since $i$ may not know $j$'s utility function, $i$ cannot predict which deviation $j$ would find profitable, and thus hedges against the worst case.
\end{enumerate}

\begin{table}[t]
  \centering
  \small
  \renewcommand{\arraystretch}{1.15}
  \setlength{\tabcolsep}{6pt}

  \begin{tabularx}{\linewidth}{@{}l *{4}{>{\centering\arraybackslash}X}@{}}
    \toprule
      & Nash
      & \makecell{(Coarse)\\Correlated}
      & \makecell{External\\Payment}
      & \makecell{SETE\\(ours)} \\
    \midrule
    Can it work without a third party (mediator)?
      & \cmarkcell & \xmarkcell & \xmarkcell & \cmarkcell \\
    Does it prescribe independent behavior?
      & \cmarkcell & \xmarkcell & \cmarkcell & \cmarkcell \\
    Is it budget-balanced?
      & \cmarkcell & \cmarkcell & \xmarkcell & \cmarkcell \\
    Can it be computed efficiently?
      & \xmarkcell & \cmarkcell & \xmarkcell~$\slash$ \cmarkcell & \cmarkcell \\
    \bottomrule
  \end{tabularx}

  \caption{Comparison among Nash equilibrium (NE) \citep{nash1950equilibrium-nash-def}, correlated equilibrium (CE) \citep{aumann1974subjectivity-CE-def} and coarse correlated equilibrium (CCE) \citep{moulin1978strategically-CCE-def}, equilibrium with external payments \citep{monderer2003k-implementation,kolumbus2024paying-do-better}, and our proposed mediator-free Self-Enforcing Transfer Equilibrium (SETE). The table focuses on SETE; the mediated extension M-SETE in \Cref{section:mediated-SETE} intentionally adds binding mediation to obtain a Nash equilibrium of the augmented game. The computational complexity of external payment methods is definition-dependent, as it varies with the particular notion of equilibrium.}
  \label{table:comparison}
\end{table}

\subsection{Contributions}

\paragraph{Equilibrium Definitions} We extend the original game to a two-stage sequential game: in Stage 1, players commit to a payment plan, and in Stage 2, they commit to strategies in the original game.

We then define \emph{Self-Enforcing Transfer Equilibrium} (SETE) in \Cref{section:concept} as a pair $\rbr{\pi,\bp}$, where $\pi$ is a strategy profile and $\bp=\rbr{p_{i\to j}}_{i\neq j}$ is a payment plan. Here, $p_{i\to j}$ denotes the payment from player $i$ to player $j$, conditional on player $j$ not deviating from $\pi$. SETE is built on two requirements: (i) no player wishes to withdraw any promised payments, and (ii) no player wishes to deviate from $\pi$.

For this mechanism, we relate SETE to equilibria of the augmented game. We show that NE of the augmented sequential game can be too strong to accommodate socially optimal strategies (\Cref{lemma:sw-optimal-not-NE}). Accordingly, \Cref{theorem:agent-normal-form} proves that every SETE induces an NE in the agent normal form of the augmented game.

\paragraph{Mediated Strengthening} In \Cref{section:mediated-SETE}, we introduce \emph{Mediated Self-Enforcing Transfer Equilibrium} (M-SETE). The mediator makes the payment schedule and prescribed strategies binding offers, so each player either accepts the entire offer or declines the mechanism. A player can no longer use payments to induce the others to follow $\pi_{-i}$ while simultaneously deviating from $\pi_i$. We prove that any M-SETE is a Nash equilibrium of the augmented game induced by this binding-offer mechanism. Moreover, in any finite game, any strategy profile whose social welfare is no smaller than that of the worst NE, and in particular any socially optimal strategy profile, can be supported as an M-SETE while maintaining budget balance (\Cref{corollary:NE-M-SETE}). Thus M-SETE provides a mediated workaround to the limitation identified for SETE.

\paragraph{Correspondence between SETE and Stationary Points of the Social Welfare Function} In polymatrix games, we show that any \emph{stationary point} of the social welfare function induces a SETE via an appropriate payment plan (\Cref{lemma:social-welfare}). This bridges social welfare optimization and equilibrium. Furthermore, this implies that any socially optimal strategy profile can be sustained as a SETE (\Cref{corollary:sw-optimal-eq}), and therefore as an NE in the agent normal form of the augmented game. With M-SETE, the corresponding socially optimal profile can instead be implemented as an NE of the augmented game.

This connection also yields existence: since the social welfare function is continuous over a compact set, it attains a maximizer, and such a maximizer is a stationary point. Hence, a SETE exists. Moreover, in the same spirit as the minimax-based existence proof for correlated equilibrium and its computational implications \citep{hart1989existence-phi-zero-sum}, our characterization suggests algorithmic routes to computing SETE that bypass the {\tt PPAD}-hardness barrier faced by NE.

\paragraph{Efficient Computation} We provide a polynomial-time algorithm (see \Cref{eq:update-rule-potential-game}) for computing SETE in polymatrix games. This tractability is noteworthy because SETE retains \emph{independent (product-form) mixed strategies}, as in NE. Unlike correlated equilibrium, SETE does not relax to correlated distributions; instead, it layers a payment plan on top of product-form play.

Since computing Nash equilibrium is {\tt PPAD}-hard even in two-player general-sum games \citep{chen2006settling-ppad,daskalakis2009complexity-PPAD}, our result highlights a computational benefit of transfers: suitably designed internal payments reshape incentives enough to make an independent-strategy equilibrium efficiently computable.

\paragraph{Decentralized Learning} Unlike prior work \citep{monderer2003k-implementation,kolumbus2024paying-do-better} that relies on a fully informed mediator, we present a decentralized learning dynamic (\Cref{alg:learning}) in which players learn an equilibrium using only observed utilities and realized payments. The update rule is natural: at each timestep, each player computes payments to others and performs one step of projected gradient ascent on (utility $+$ payments received).

\begin{figure}[t]
\centering
\begin{minipage}[t]{0.48\textwidth}
    \centering
    \vspace{-10em}
    \begin{tabular}{|c|c|c|}
    \hline
     & \textbf{Cooperate (C)} & \textbf{Defect (D)} \\
    \hline
    \textbf{Cooperate (C)} & $(0.6,0.6)$ & $(0,1)$ \\
    \hline
    \textbf{Defect (D)} & $(1,0)$ & $(0.2,0.2)$ \\
    \hline
    \end{tabular}
    \captionof{table}{Utility matrix for the Prisoner's Dilemma. Each entry $(a,b)$ gives the payoffs to the row player ($a$) and the column player ($b$).}
    \label{table:PD}
\end{minipage}\hfill
\begin{minipage}[t]{0.48\textwidth}
    \centering
    \includegraphics[width=\linewidth]{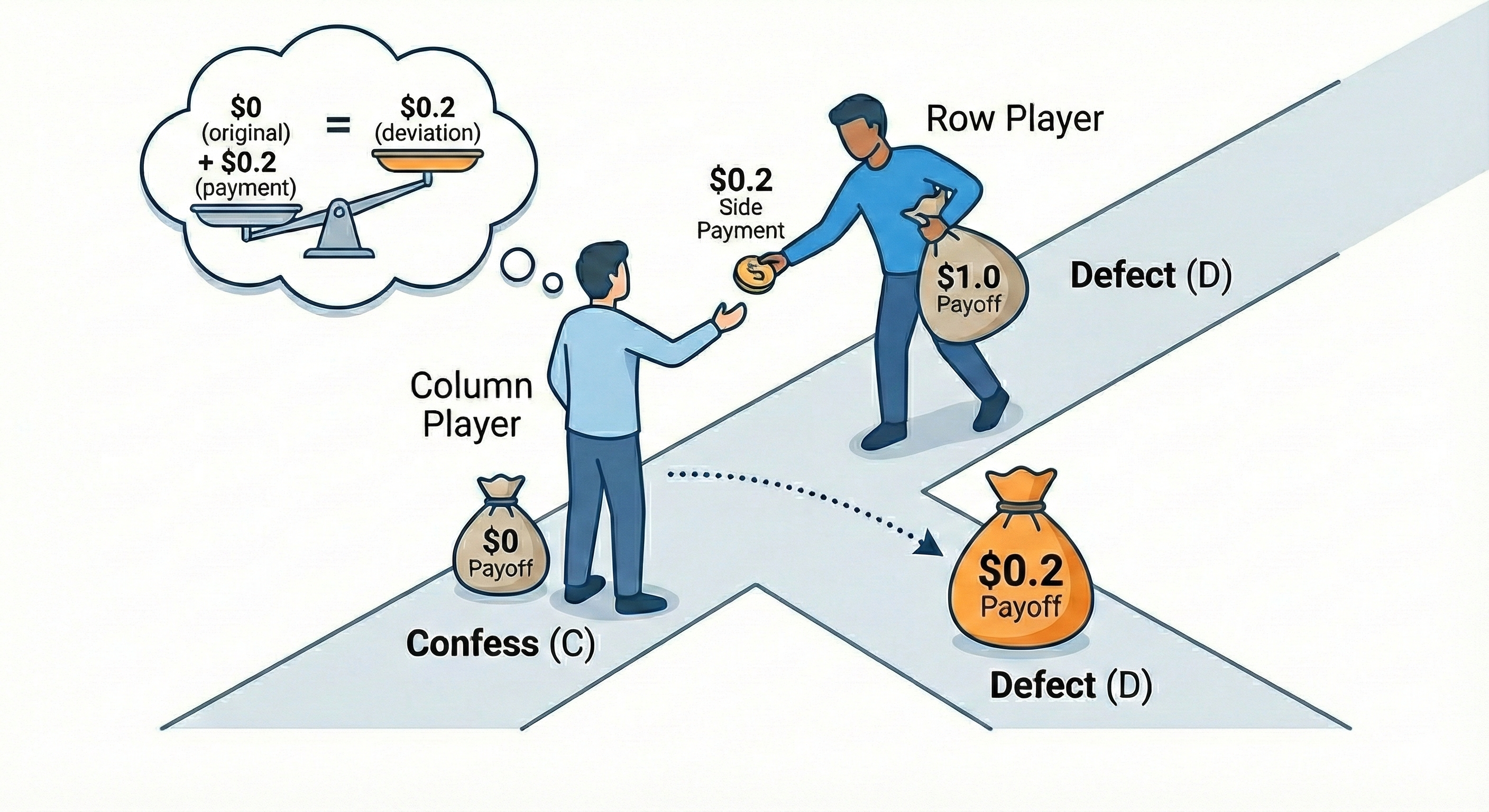}
    \caption{An illustration of how side payments work in Prisoner's Dilemma.}
    \label{fig:PD}
\end{minipage}

\end{figure}

\section{Related Work}
\label{section:related-work}

\paragraph{Endogenous Game.} The works most closely related to this paper are \citet{jackson2005endogenous} and \citet{kolumbus2024paying-do-better}, which explore stronger payment mechanisms where players can offer payments contingent on the joint action executed by all players. Under their framework, players can arbitrarily modify the utility function as long as the sum of the utility functions is maintained. However, this approach causes the decision variables to scale exponentially with the number of players $N$. In contrast, our mechanism's decision variables scale polynomially with $N$ because payments depend only on the sender and receiver. Furthermore, while previous models assume players possess full knowledge of others' utility functions, we provide an interaction principle to learn payments without such restrictive assumptions. Finally, whereas \citet{jackson2005endogenous} and \citet{kolumbus2024paying-do-better} focus on equilibrium properties, we extend the study to the decentralized learning of these equilibria under internal transfers.

\paragraph{Pigouvian Tax.} As demonstrated by \citet{roughgarden2002bad-POA}, NEs in certain congestion games often result in social welfare strictly lower than the maximal social welfare. To address this, \citet{sandholm2002evolutionary-pigouvian, sandholm2005negative-pigouvian} proposed the Pigouvian tax, additional payments imposed on specific routes, which transforms the socially optimal strategy into an NE. While this mechanism improves welfare, it relies on external intervention. Moreover, Pigouvian taxes are typically studied in potential games \citep{monderer1996potential-game}, where players' utilities derive from a common potential function (up to player-dependent constants). In contrast, we consider general-sum games.

\paragraph{Exogenous Payment.} Another research thread explores non-negative external payments to incentivize strategies with higher social welfare. \citet{monderer2003k-implementation} analyzed the minimum external payment required to make a strategy dominant. Similarly, \citet{zhang2023steering} introduced an online learning setting where a mediator provides iterative payments to steer no-regret learners toward a desired NE. Unlike these works, we focus exclusively on internal transfers, ensuring the system remains budget-balanced with zero net payments to or from external entities.

\section{Preliminaries}

\paragraph{Basics.} For any vector $\bx\in\RR^n$, let $x_i$ denote its $i^{th}$ coordinate and let $\nbr{\bx}_p$ denote its $\ell_p$-norm. By default, we write $\nbr{\bx}\coloneqq \nbr{\bx}_2$. For any $\bx,\by\in\RR^n$, define the inner product $\inner{\bx}{\by}\coloneqq \sum_{i=1}^n x_i\cdot y_i$. For simplicity, let $\zero$ and $\one$ denote the all-zero and all-one vectors, respectively. For any positive integer $N$, let $[N]\coloneqq \cbr{1,2,\dots,N}$ and $\Delta^N\coloneqq \cbr{\bx\in [0,1]^N \given \sum_{i=1}^N x_i=1}$ denote the $(N-1)$-dimensional probability simplex. More generally, for any finite discrete set $S$, let $|S|$ denote its cardinality and let $\Delta^S$ denote the $(|S|-1)$-dimensional probability simplex over $S$, indexed by elements of $S$.

Let $\Omega \subseteq \RR^n$ be a nonempty closed convex set and let $\bx\in \RR^n$. Define the (Euclidean) projection of $\bx$ onto $\Omega$ by
\begin{align}
    \Proj{\Omega}{\bx}
\coloneqq
\argmin_{\bx'\in\Omega}\ \frac{1}{2}\nbr{\bx' - \bx}^2. \numberthis[Projection]{eq:projection}
\end{align}
For any differentiable function $f:\Omega\to \RR$, we call a vector $\bx\in\Omega$ an \emph{$\epsilon$-stationary point} if, for every $\lambda>0$,
\begin{align}
    \frac{1}{\lambda}\nbr{\Proj{\Omega}{\bx + \lambda \nabla f(\bx)} - \bx} \leq \epsilon. \numberthis[$\epsilon$-Stationary Point]{eq:stationary-point}
\end{align}

\paragraph{Games.} A (finite) game is a tuple $\rbr{N, \cbr{\cA_i}_{i=1}^N, \cbr{\cU_i}_{i=1}^N}$, where:
\begin{itemize}
    \item Integer $N$ is the number of players.
    \item Finite discrete set $\cA_i$ is the action set of player $i$. Let $\cA \coloneqq \bigtimes_{i=1}^N \cA_i$ denote the joint action set, and let $A \coloneqq \max_{i\in [N]} |\cA_i|$ denote the size of the largest action set.
    \item Function $\cU_i\colon \cA \to [0,1]$ is the utility function of player $i$, where $\cU_i(\ba)$ is the utility received by player $i$ under joint action $\ba=(a_1,a_2,\dots,a_N)\in\cA$.
\end{itemize}
For any $\ba\in\cA$, let $\ba_{-i}\coloneqq (a_1,\dots,a_{i-1},a_{i+1},\dots,a_N)$ denote the action profile of all players except $i$, and define $\cA_{-i}\coloneqq \bigtimes_{j\in [N]\setminus\cbr{i}} \cA_j$.

For each player $i\in [N]$, a (mixed) strategy $\pi_i\in \Delta^{\cA_i}$ is a probability distribution over $\cA_i$. A strategy profile is $\pi=(\pi_1,\pi_2,\dots,\pi_N)\in \bigtimes_{i=1}^N \Delta^{\cA_i}$. Under $\pi$, the expected utility of player $i$ is
\begin{align}
    \EE_{\ba\sim \pi}\sbr{\cU_i(\ba)}=\sum_{\ba\in\cA} \cU_i(\ba)\prod_{j=1}^N \pi_j(a_j).
\end{align}

\paragraph{Polymatrix Games.} In a polymatrix game, each player's utility decomposes as a sum of pairwise interactions:
\begin{align}
    \forall i\in [N], \ba\in\cA,\quad \cU_i(\ba)=\sum_{j\in\cN(i)} \cU_{i, j}(a_i, a_j),
\end{align}
where $\cN(i)$ denotes the set of neighbors of player $i$, and $\cU_{i,j}\colon \cA_i\times \cA_j \to [0,1]$ is the utility function for the two-player game between $i$ and neighbor $j\in \cN(i)$. Additionally, without loss of generality, we assume $\cU_i(\ba)\in [0, 1]$ for every player $i\in[N]$ and joint action $\ba\in\cA$. Thus, a polymatrix game models settings in which each player engages in separate pairwise interactions, and their total utility is the sum of the corresponding pairwise payoffs.

\paragraph{Agent Normal Form.} A sequential game is specified by a game tree with a set of decision points $s\in\cS$. Let $p(s)$ denote the original player who moves at $s$ and $\cA_s$ be the set of actions available at $s$.

The agent normal form of a sequential game \citep{selten1975reexamination-agent-form} is a normal-form game, in which each decision point $s\in\cS$ is controlled by a distinct agent (player). A joint action $\ba\in \prod_{s\in\cS}\cA_s$ specifies one action at every decision point and thereby induces an outcome in the original game, together with utilities $\cU_i(\ba)$ for each original player $i$. The agent normal form is then the game $\rbr{\cS,\cbr{\cA_s}_{s\in\cS},\cbr{u_s}_{s\in\cS}}$ with payoffs
\begin{align*}
    u_s(\ba) \coloneqq \cU_{p(s)}(\ba)\qquad \text{for all } s\in\cS,
\end{align*}
so that all agents corresponding to the same original player share the same utility function.

\section{Self-Enforcing Transfer Equilibrium (SETE)}
\label{section:concept}

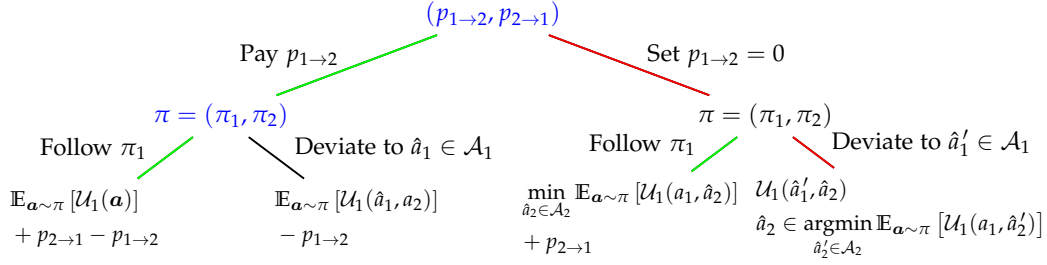
\begin{figure}[t]
\centering
\scalebox{0.9}{
\begin{tikzpicture}[
  level distance=1.5cm,
  level 1/.style={sibling distance=8cm},
  level 2/.style={sibling distance=4cm},
  level 3/.style={sibling distance=2cm},
  edge from parent/.style={draw,thick}
]
\node (root) {\color{blue}$\rbr{p_{1\to 2}, p_{2\to 1}}$}
  child { node (A) {\color{blue}$\pi=\rbr{\pi_1,\pi_2}$}
    child { node (B) {\small$
            \begin{aligned}
            &\EE_{\ba\sim\pi}\sbr{\cU_1(\ba)}\\
            &+ p_{2\to 1} - p_{1\to 2}
            \end{aligned}$}
            edge from parent node[midway, yshift=4pt, left=4pt] {Follow $\pi_1$}
        }
    child { node (C) {\small$
            \begin{aligned}
                &\EE_{\ba\sim\pi}\sbr{\cU_1(\hat a_1, a_2)} \\ 
                &- p_{1\to 2}
            \end{aligned}$}
            edge from parent node[midway, yshift=4pt, right=4pt] {Deviate to $\hat a_1\in\cA_1$}
        }
    edge from parent node[midway, yshift=4pt, left=4pt] {Pay $p_{1\to 2}$}
  }
  child { node (D) {$\pi=\rbr{\pi_1,\pi_2}$}
    child { node (E) {\small$
            \begin{aligned}
            &\min_{\hat a_2\in\cA_2}\EE_{\ba\sim\pi}\sbr{\cU_1(a_1, \hat a_2)}\\
            &+ p_{2\to 1}
            \end{aligned}$}
            edge from parent node[midway, yshift=4pt, left=4pt] {Follow $\pi_1$}
        }
    child { node (F) {\small$
            \begin{aligned}
            &\cU_1(\hat a_1^{\,\prime}, \hat a_2)\\
            &\hat a_2\in \argmin_{\hat a_2^{\,\prime}\in\cA_2}\EE_{\ba\sim\pi}\sbr{\cU_1(a_1, \hat a_2^{\,\prime})}
            \end{aligned}$} 
            edge from parent node[midway, yshift=4pt, right=4pt] {Deviate to $\hat a_1^{\,\prime}\in\cA_1$}
        }
    edge from parent node[midway, yshift=4pt, right=4pt] {Set $p_{1\to 2}=0$}
  };

\draw[green,thick] (root) -- (A);
\draw[green,thick] (A) -- (B);
\draw[red,thick] (root) -- (D);
\draw[green,thick] (D) -- (E);
\draw[red,thick] (D) -- (F);
\end{tikzpicture}
}
\caption{Illustration of the augmented two-stage game, taking player 1 as the player of interest. In Stage 1, players choose a payment plan; in Stage 2, they choose strategies in the original game. The green paths form a Nash equilibrium of the induced agent normal form of the augmented game, but this equilibrium need not be a Nash equilibrium of the original sequential (two-stage) game. Accordingly, the red paths depict a potentially profitable deviation in the original sequential game, even though the green paths are stable in the agent normal form. This limitation motivates the mediated binding-offer mechanism in \Cref{section:mediated-SETE}. Blue nodes are on-path decision points. Without loss of generality, deviations can be taken to be pure strategies, since expected utilities are linear in mixed strategies.}
\label{fig:augmented-game}
\end{figure}

First, we study equilibria in an augmented game (see \Cref{fig:augmented-game}) in which, beyond choosing actions, players may commit to \emph{internal} (peer-to-peer) transfers. The augmented game is defined as follows.
\begin{definition}[Augmented Game with Payments]
    \begin{enumerate}[left=0mm,nosep,leftmargin=*, align=left, labelsep=0pt, labelwidth=0pt,label={},ref={Augmented Game}] \item
    The augmented sequential game consists of two stages. In Stage 1, each player $i\in [N]$ commits to a payment plan $\rbr{p_{i\to j}}_{j\in [N]\setminus\cbr{i}}$. In Stage 2, each player $i\in [N]$ chooses a mixed strategy $\pi_i\in\Delta^{\cA_i}$. \label{definition:augmented-game}
    \end{enumerate}
\end{definition}
Note that \Cref{definition:augmented-game} follows the two-stage framework of \citet{jackson2005endogenous}.

Importantly, these transfers are \emph{budget-balanced}: they only redistribute utility among players and do not inject or destroy value in the system. Concretely, for a strategy profile $\pi=\rbr{\pi_i}_{i=1}^N$, each player $i\in[N]$ can commit to a nonnegative payment $p_{i\to j}\geq 0$ to player $j$ conditional on $j$ not deviating from the prescribed strategy $\pi_j$.\footnote{Equivalently, one can view $p_{i\to j}$ as being placed in escrow and released to $j$ only upon compliance. Since compliance is checked after play, a player cannot ``take the money and then deviate'' in the same interaction.} When no player deviates and all prescribed payments are executed, then player $i$'s net payoff equals
\begin{align*}
    \EE_{\ba\sim\pi}\sbr{\cU_i\rbr{\ba}}
    \;-\;\sum_{\substack{j=1,\\j\neq i}}^N p_{i\to j}
    \;+\;\sum_{\substack{j=1,\\j\neq i}}^N p_{j\to i},
\end{align*}
and summing the payments $\bp$ over $i$ shows that the transfers are purely redistributive:
\begin{align*}
    \sum_{i=1}^N \rbr{
    -\sum_{j\neq i} p_{i\to j} + \sum_{j\neq i} p_{j\to i}}
    \;=\;0,
\end{align*}
so the social welfare $\sum_{i=1}^N \EE_{\ba\sim\pi}\sbr{\cU_i\rbr{\ba}}$ is unchanged by transfers.

The strategy profile and the internal transfers should satisfy two criteria: stability and incentive compatibility.

\subsection{Stability} 

To ensure the stability of $\rbr{\pi_i}_{i=1}^N$ (no player wants to deviate), we require that, for any player $i\in[N]$,
\begin{align}
    \max_{\hat a_i\in\cA_i} \EE_{\ba\sim\pi} \sbr{\cU_i\rbr{\hat a_i, \ba_{-i}}}
    \;\leq\;
    \EE_{\ba\sim\pi} \sbr{\cU_i\rbr{\ba}} + \sum_{\substack{j=1,\\j\neq i}}^N p_{j\to i},
    \numberthis[Stability]{eq:stable}
\end{align}
where $\ba\sim\pi$ means that $a_i\sim \pi_i$ is sampled independently for each $i\in[N]$. \Cref{eq:stable} captures that a unilateral deviation by $i$ forfeits the incoming transfers $\sum_{j\neq i}p_{j\to i}$, while $i$'s outgoing commitments $\sum_{j\neq i}p_{i\to j}$ do not affect the comparison (they are independent of $i$'s own deviation and thus cancel from both sides). When there is no internal transfer, \emph{i.e.}, $p_{i\to j}\equiv 0$ for all $i,j\in[N]$, \Cref{eq:stable} reduces to the standard definition of Nash equilibrium \citep{nash1950equilibrium-nash-def}.

\subsection{Incentive Compatibility}

Beyond stability, we also require that each player has the incentive to \emph{honor} the proposed transfers, rather than unilaterally withdrawing from the transfer system (equivalently, setting some outgoing payments to zero). Intuitively, player $i$ is willing to pay others only insofar as doing so prevents \emph{costly deviations} that would hurt $i$. Hence, $i$'s total payment to any group of recipients should be bounded by the (worst-case) loss that $i$ would suffer if those recipients were to deviate from the prescribed strategy profile $\pi$. Formally, for any player $i\in[N]$ and any subset $S_i\subseteq [N]\setminus\cbr{i}$,
\begin{align}
    \sum_{j\in S_i} p_{i\to j}
    \;\leq\;
    \EE_{\ba\sim\pi}\sbr{\cU_i\rbr{\ba}}
    \;-\;
    \min_{\hat\ba_{S_i}\in\cA_{S_i}}
    \EE_{\ba\sim\pi}\sbr{\cU_i\rbr{\hat\ba_{S_i}, \ba_{-S_i}}}.
    \numberthis[Compatibility]{eq:compatible}
\end{align}
\Cref{eq:compatible} ensures that player $i$'s payment to any subset $S_i$ does not exceed $i$'s \emph{willingness to pay}: the maximum utility loss that $i$ would incur if the players in $S_i$ were no longer incentivized to follow $\pi$. In particular, if \Cref{eq:compatible} fails for some $S_i$, then player $i$ would strictly benefit from withdrawing those transfers, \emph{i.e.}, setting $p_{i\to j}=0$ for all $j\in S_i$. Finally, \Cref{eq:compatible} is well-defined since the right-hand side is always nonnegative: the minimization over $\hat\ba_{S_i}$ is at most the value achieved when $S_i$ follows $\pi_{S_i}$, and therefore cannot exceed $\EE_{\ba\sim\pi}\sbr{\cU_i\rbr{\ba}}$.

Hence, we define the Self-Enforcing Transfer Equilibrium (SETE) as follows.

\begin{definition}[Self-Enforcing Transfer Equilibrium (SETE)]
\label{def:SETE-def-label}
    \begin{enumerate}[left=0mm,nosep,leftmargin=*, align=left, labelsep=0pt, labelwidth=0pt,label={},ref={SETE}] \item
        In any game, a Self-Enforcing Transfer Equilibrium (SETE) is a strategy--payment pair $\rbr{\pi,\bp}$ that satisfies \Cref{eq:stable} and \Cref{eq:compatible} simultaneously, where $\pi=\rbr{\pi_i}_{i=1}^N$ is a strategy profile and $\bp=\rbr{p_{i\to j}}_{i\neq j\in[N]}$ is a payment plan. Here, $p_{i\to j}$ is the (nonnegative) payment promised by player $i$ to player $j$, executed conditional on player $j$ not deviating from $\pi_j$. \label{def:SETE}
    \end{enumerate}
\end{definition}

Note that \Cref{def:SETE} always exists, as stated in the following lemma.
\begin{lemma}
\label{lemma:existence}
    For any finite game, \Cref{def:SETE} always exists. Moreover, for any Nash equilibrium $\pi^{\rm NE}$, the pair $(\pi^{\rm NE},\zero)$ is a \Cref{def:SETE}.
\end{lemma}
\begin{proof}
For any NE $\pi^{\rm NE}$, the pair $\rbr{\pi^{\rm NE}, \zero}$ satisfies both \Cref{eq:stable} and \Cref{eq:compatible}. Moreover, since every finite game admits a NE \citep{nash1950equilibrium-nash-def}, \Cref{def:SETE} always exists as well.
\end{proof}
Finally, we define the approximate SETE in the following.

\begin{definition}
[$\epsilon$-approximate SETE]
\begin{enumerate}[left=0mm,nosep,leftmargin=*, align=left, labelsep=0pt, labelwidth=0pt,label={},ref={$\epsilon$-approximate SETE}]
\item
For any $\epsilon\geq 0$, a strategy--payment pair $(\pi,\bp)$ is an $\epsilon$-approximate \Cref{def:SETE} if and only if it satisfies \Cref{eq:compatible} and the following $\epsilon$-approximate \Cref{eq:stable} constraint: for any $i\in[N]$,
\label{def:approximate-SETE}
\begin{align}
    \max_{\hat a_i\in\cA_i} \EE_{\ba\sim\pi} \sbr{\cU_i\rbr{\hat a_i, \ba_{-i}}} \leq \EE_{\ba\sim\pi} \sbr{\cU_i\rbr{\ba}} + \sum_{\substack{j=1,\\j\neq i}}^N p_{j\to i} + {\color{red} \epsilon}. \label{eq:approximate-stability}
\end{align}
\end{enumerate}
\end{definition}
Note that $(\pi,\zero)$ being an \Cref{def:approximate-SETE} is equivalent to $\pi$ being an $\epsilon$-approximate Nash equilibrium \citep{algorithmic-game-theory}.

\subsection{Connection to Equilibria in the Augmented Game}

Next, we relate \Cref{def:SETE} to NE in the agent normal form \citep{selten1975reexamination-agent-form} of \Cref{definition:augmented-game}.

As illustrated in \Cref{fig:augmented-game}, fix a player $i$ and a subset of recipients $S_i$. If player $i$ unilaterally reduces or withdraws the payments to the recipients in $S_i$ in Stage 1, the continuation behavior in the corresponding Stage 2 subgame is specified as follows: the players in $S_i$ coordinate on a joint action that is worst for player $i$ in expectation under $\pi$,
\begin{align*}
    \hat \ba_{S_i}\in \argmin_{\hat\ba_{S_i}^{\,\prime}\in\bigtimes_{j\in S_i} \cA_j} \EE_{\ba\sim\pi}\sbr{\cU_i\rbr{\hat\ba_{S_i}^{\,\prime}, \ba_{-S_i}}},
\end{align*}
while players outside $S_i$ continue according to their prescribed strategies.

Now consider the on-path history in which no player withdraws any payment (the left child of the root in \Cref{fig:augmented-game}). By \Cref{eq:stable}, given that Stage 1 ends with no withdrawals, no player can profitably deviate in Stage 2. Moreover, \Cref{eq:compatible} ensures that, for every player $i$, withdrawing payments to trigger a withdrawal subgame is not strictly profitable when the continuation in that subgame is given by the SETE specification above (hence no player prefers to move to the right child of the root in \Cref{fig:augmented-game}). Therefore, no player has a profitable deviation at any decision node reached with positive probability under SETE (the blue nodes in \Cref{fig:augmented-game}).

However, \Cref{eq:stable} and \Cref{eq:compatible} do not by themselves preclude profitable deviations at off-path decision nodes (the right child of the root in \Cref{fig:augmented-game}). In particular, a player might profit by jointly deviating: withdrawing payments in Stage 1 and then deviating again in Stage 2. Thus, these two constraints do not, in general, guarantee that \Cref{def:SETE} is an NE of the original augmented sequential game.

Indeed, \Cref{lemma:sw-optimal-not-NE} shows that there exists a two-player game in which the socially optimal strategy profile cannot be supported as a Nash equilibrium of the augmented game by any payment plan.
\begin{lemma}
\label{lemma:sw-optimal-not-NE}
    There exists a two-player game whose unique socially optimal strategy profile cannot be supported by any payment plan as a Nash equilibrium of \Cref{definition:augmented-game}.
\end{lemma}
The proof is postponed to \Cref{section:concept-appendix}.
Consequently, rather than requiring an NE in the augmented sequential game, we turn to the weaker solution concept: the NE in the agent normal form of the augmented game. As \Cref{theorem:agent-normal-form} shows, SETE induces this weaker equilibrium notion. In \Cref{section:mediated-SETE}, we then show that, when a mediator can make both the payment schedule and the prescribed strategies binding offers to each player, the resulting notion yields an NE of the augmented sequential game itself.
\begin{theorem}
\label{theorem:agent-normal-form}
    Every \Cref{def:SETE} is a Nash equilibrium in the agent normal form of the \Cref{definition:augmented-game}.
\end{theorem}
\begin{proof}
    In the agent normal form, each decision node is controlled by a \emph{distinct} agent who shares the utility of the corresponding original player. The completed profile prescribes the SETE payment plan in Stage 1 and the SETE strategy profile in the on-path Stage 2 game. By \Cref{eq:stable}, no on-path Stage 2 agent has a profitable deviation. By \Cref{eq:compatible}, no Stage 1 payment agent can profit by reducing or withdrawing any subset of outgoing payments, because the resulting continuation loss is at least as large as the payments saved. Increasing payments is weakly worse because it raises outgoing transfers without improving the prescribed continuation. Agents controlling unreached off-path decision nodes cannot change the outcome through a unilateral deviation in the agent normal form, since the path to their nodes is controlled by other agents and has probability zero under the completed profile. Therefore no agent has a profitable unilateral deviation, and the completed profile is a Nash equilibrium of the agent normal form. \qedhere
\end{proof}
Finally, \Cref{theorem:agent-normal-form} does not contradict the negative results of \citet{jackson2005endogenous}, who show that, under a different definition of payment plans (see \Cref{section:related-work}), socially optimal strategy profiles need not be Nash equilibria of the augmented game in certain two-player instances. Nor does it contradict the M-SETE result in \Cref{section:mediated-SETE}, which changes the mechanism by adding binding offers.

\section{Efficient Computation of \Cref{def:SETE} in Polymatrix Games}
\label{section:polymatrix}

In this section, we establish three useful properties of \Cref{def:SETE} in polymatrix games:
\begin{itemize}
    \item A decomposable \Cref{eq:compatible} constraint (\Cref{lemma:decomposition-compatible}).
    \item The connection of \Cref{def:SETE} to the social welfare function (\Cref{lemma:social-welfare}, \Cref{example:PD-SETE}).
    \item Efficient computation of equilibria under transfers, circumventing the \texttt{PPAD}-hardness of Nash equilibrium computation (\Cref{theorem:convergence-potential-game}).
\end{itemize}

\subsection{Decomposable \Cref{eq:compatible} Constraint}

We begin by showing that, in polymatrix games, the subset-based compatibility requirement in \Cref{eq:compatible} reduces to a collection of independent, pairwise constraints. Concretely, for any player $i\in[N]$, rather than verifying \Cref{eq:compatible} for all subsets $S_i\subseteq [N]\setminus\cbr{i}$, it suffices to check a separate inequality for each potential recipient $j\neq i$.

\begin{lemma}
\label{lemma:decomposition-compatible}
Consider a polymatrix game and a nonnegative payment plan $\bp$. Then \Cref{eq:compatible} is satisfied if and only if for any player $i\in [N]$ and any $j\in [N]\setminus\cbr{i}$,
\begin{align}
    p_{i\to j}
    \leq
    \EE_{\ba\sim\pi}\sbr{\cU_i\rbr{\ba}}
    - \min_{\hat a_j\in\cA_j}\EE_{\ba\sim\pi}\sbr{\cU_i\rbr{\hat a_j, \ba_{-j}}}.
    \label{eq:independent-compatible}
\end{align}
\end{lemma}

By \Cref{lemma:decomposition-compatible}, verifying \Cref{eq:compatible} no longer requires enumerating exponentially many subsets $S_i$. Instead, it reduces to checking the $N-1$ pairwise constraints in \Cref{eq:independent-compatible} for each fixed $i$. This decomposition is a key ingredient behind the efficient computation results for \Cref{def:SETE} in polymatrix games.

\begin{proof}
    \textbf{\Cref{eq:compatible} implies \Cref{eq:independent-compatible}.}\quad When \Cref{eq:compatible} is satisfied, by taking $S_i=\cbr{j}$ for each $j\in [N]\setminus\cbr{i}$, \Cref{eq:independent-compatible} is satisfied. Hence, \Cref{eq:compatible} implies \Cref{eq:independent-compatible}.

    \textbf{\Cref{eq:independent-compatible} implies \Cref{eq:compatible}.}\quad When \Cref{eq:independent-compatible} is satisfied, then for any subset $S_i\subseteq[N]\setminus\cbr{i}$, we have
    \begin{align*}
        \sum_{j\in S_i} p_{i\to j}\leq& \sum_{j\in S_i} \rbr{\EE_{\ba\sim\pi}\sbr{\cU_i\rbr{\ba}} - \min_{\hat a_j\in\cA_j} \EE_{\ba\sim\pi}\sbr{\cU_i\rbr{\hat a_j, \ba_{-j}}}}\\
        \overset{(i)}{=}& \sum_{j\in S_i\cap\cN(i)} \rbr{\EE_{\ba\sim\pi}\sbr{\cU_i\rbr{\ba}} - \min_{\hat a_j\in\cA_j} \EE_{\ba\sim\pi}\sbr{\cU_i\rbr{\hat a_j, \ba_{-j}}}}.
    \end{align*}
    $(i)$ is because $\cU_i(\ba)=\sum_{k\in\cN(i)} \cU_{i, k}(a_i, a_k)$ by definition of the polymatrix game. Then,
    \begin{align*}
        \min_{\hat a_j\in\cA_j} \EE_{\ba\sim\pi}\sbr{\cU_i\rbr{\hat a_j, \ba_{-j}}} =& \EE_{\ba\sim\pi}\sbr{\cU_i\rbr{\ba}}
    \end{align*}
    if $j\not\in\cN(i)$, since $\cU_i$ is irrelevant to the action of player $j\notin \cbr{i}\cup \cN(i)$. Furthermore, for any $j\in S_i\cap \cN(i)$, we have
    \begin{align*}
        \min_{\hat a_j\in\cA_j} \EE_{\ba\sim\pi}\sbr{\cU_i\rbr{\hat a_j, \ba_{-j}}}=& \min_{\hat a_j\in\cA_j} \rbr{\EE_{\ba\sim\pi}\sbr{\cU_{i, j}\rbr{a_i, \hat a_j}} + \sum_{k\in \cN(i)\setminus\cbr{j}} \EE_{\ba\sim\pi}\sbr{\cU_{i, k}\rbr{a_i, a_k}} }\\
        =& \min_{\hat a_j\in\cA_j} \EE_{\ba\sim\pi}\sbr{\cU_{i, j}\rbr{a_i, \hat a_j}} + \sum_{k\in \cN(i)\setminus\cbr{j}} \EE_{\ba\sim\pi}\sbr{\cU_{i, k}\rbr{a_i, a_k}}.
    \end{align*}
    Also, note that
    \begin{align*}
        &\min_{\hat\ba_{S_i}\in\cA_{S_i}} \EE_{\ba\sim\pi}\sbr{\cU_i\rbr{\hat\ba_{S_i}, \ba_{-S_i}}}\\
        =&\min_{\hat\ba_{S_i}\in\cA_{S_i}} \rbr{\sum_{k\in S_i\cap\cN(i)}\EE_{\ba\sim\pi}\sbr{\cU_{i, k}\rbr{a_i, \hat a_k}} + \sum_{k\in \cN(i)\setminus S_i} \EE_{\ba\sim\pi}\sbr{\cU_{i, k}\rbr{a_i, a_k}} }\\
        =&\sum_{k\in S_i\cap\cN(i)} \min_{\hat a_k\in\cA_k} \EE_{\ba\sim\pi}\sbr{\cU_{i, k}\rbr{a_i, \hat a_k}} + \sum_{k\in \cN(i)\setminus S_i} \EE_{\ba\sim\pi}\sbr{\cU_{i, k}\rbr{a_i, a_k}}.
    \end{align*}
    Therefore,
    \begin{align*}
        \sum_{j\in S_i} p_{i\to j}\leq& \sum_{j\in S_i\cap\cN(i)} \sbr{\EE_{\ba\sim\pi}\sbr{\cU_{i, j}\rbr{a_i, a_j}} - \min_{\hat a_j\in\cA_j} \EE_{\ba\sim\pi}\sbr{\cU_{i, j}\rbr{a_i, \hat a_j}}}\\
        \leq& \EE_{\ba\sim\pi}\sbr{\cU_i\rbr{\ba}} - \min_{\hat\ba_{S_i}\in\cA_{S_i}} \EE_{\ba\sim\pi}\sbr{\cU_i\rbr{\hat\ba_{S_i}, \ba_{-S_i}}}. \qedhere
    \end{align*}
\end{proof}

\subsection{Connection to the Social Welfare}

We further show that, once internal transfers are allowed, any strategy profile that maximizes social welfare, \emph{i.e.}, the sum of all players' expected utilities, can be sustained as a \Cref{def:SETE}.

\begin{lemma}
\label{lemma:social-welfare}
Consider any polymatrix game and fix $\epsilon \geq 0$. Let $\pi$ be an \Cref{eq:stationary-point} of the social welfare function $SW$, where
\begin{align*}
    SW \colon \bigtimes_{i=1}^N \Delta^{\cA_i} &\to [0,N], \qquad SW(\pi') \coloneqq \sum_{i=1}^N \EE_{\ba\sim\pi'}\sbr{\cU_i(\ba)}.
\end{align*}
Then there exists a payment plan $\bp$ such that $\rbr{\pi,\bp}$ is a $3\epsilon$-approximate SETE.
\end{lemma}

\Cref{lemma:social-welfare} establishes a direct link between social welfare maximization and equilibrium with internal transfers: any approximate stationary point of $SW$ can be supported (via suitable payments) as an approximate equilibrium.  The proof of \Cref{lemma:social-welfare} is as follows.

\begin{proof}
    We first introduce \Cref{lemma:variation-to-exploitability}, which connects an approximate stationary point of the social welfare to an approximate Nash equilibrium of the modified game
    $\rbr{N, \cbr{\cA_i}_{i=1}^N, \cbr{\sum_{j=1}^N \cU_j}_{i=1}^N}$, \emph{i.e.}, the game identical to the original one except that every player's utility equals the social welfare.
    \begin{lemma}
\label{lemma:variation-to-exploitability}
For any strategy profile $\pi\in\bigtimes_{i=1}^N \Delta^{\cA_i}$ and any constant $0<\lambda\leq 1/(NA)$, let
\begin{align*}
    \pi'=\Proj{\bigtimes_{i=1}^N \Delta^{\cA_i}}{\pi + \lambda \nabla SW(\pi)},
\end{align*}
where $A=\max_{i\in[N]}|\cA_i|$. Then
\begin{align*}
    \max_{\substack{i\in [N],\\\hat a_i\in\cA_i}}~~ \sum_{j=1}^N \EE_{\ba\sim \pi} \sbr{\cU_j\rbr{\hat a_i, \ba_{-i}}} - SW(\pi)\leq \frac{3}{\lambda}\sqrt{\sum_{j=1}^N \nbr{\pi_j' - \pi_j}^2}.
\end{align*}
\end{lemma}
The proof is postponed to \Cref{section:concept-appendix}.

    By \Cref{lemma:decomposition-compatible}, it suffices to bound the violation of \Cref{eq:stable}, since \Cref{eq:compatible} holds automatically once we set, for every pair $i\neq j$,
    \begin{align*}
        p_{i\to j} = \EE_{\ba\sim\pi}\sbr{\cU_i\rbr{\ba}} - \min_{\hat a_j\in\cA_j} \EE_{\ba\sim\pi}\sbr{\cU_i\rbr{\hat a_j, \ba_{-j}}},
    \end{align*}
    Fix any player $i\in[N]$. When $\pi$ is an \Cref{eq:stationary-point}, by \Cref{lemma:variation-to-exploitability}, we have
    \begin{align}
        & \max_{\hat a_i\in\cA_i} \EE_{\ba\sim\pi} \sbr{\sum_{j=1}^N \cU_j\rbr{\hat a_i, \ba_{-i}}} \leq SW(\pi) + 3\epsilon = 3\epsilon + \sum_{j=1}^N \EE_{\ba\sim\pi} \sbr{\cU_j\rbr{\ba}}.\label{eq:sw-1}
    \end{align}
    Then,
    \begin{align*}
        \EE_{\ba\sim\pi} \sbr{\cU_i\rbr{\ba}} + \sum_{j\neq i} p_{j\to i}=& \EE_{\ba\sim\pi} \sbr{\cU_i\rbr{\ba}} + \sum_{j\neq i} \rbr{\EE_{\ba\sim\pi}\sbr{\cU_j\rbr{\ba}} - \min_{\hat a_i\in\cA_i} \EE_{\ba\sim\pi}\sbr{\cU_j\rbr{\hat a_i, \ba_{-i}}}}\\
        =& \sum_{j=1}^N \EE_{\ba\sim\pi} \sbr{\cU_j\rbr{\ba}} - \sum_{j\neq i} \min_{\hat a_i\in\cA_i} \EE_{\ba\sim\pi}\sbr{\cU_j\rbr{\hat a_i, \ba_{-i}}}\\
        \overset{(i)}{\geq}& \max_{\hat a_i\in\cA_i} \sum_{j=1}^N \EE_{\ba\sim\pi} \sbr{\cU_j\rbr{\hat a_i, \ba_{-i}}} - \sum_{j\neq i} \min_{\hat a_i^{\,\prime}\in\cA_i} \EE_{\ba\sim\pi}\sbr{\cU_j\rbr{\hat a_i^{\,\prime}, \ba_{-i}}} - 3\epsilon\\
        \overset{(ii)}{\geq}& \max_{\hat a_i\in\cA_i} \EE_{\ba\sim\pi} \sbr{\cU_i\rbr{\hat a_i, \ba_{-i}}} - 3\epsilon.
    \end{align*}
    $(i)$ follows directly from \Cref{eq:sw-1}. $(ii)$ is because for any function $f,g$, we have
    \begin{align*}
        \max_x\rbr{f(x)+g(x)} \geq f(x^{f, *}) + g(x^{f, *}) \geq \max_x f(x) + \min_{x'} g(x'),
    \end{align*}
    where $x^{f, *}$ is the maximizer of $f(x)$.
    
    Since $i\in[N]$ is arbitrary, \Cref{eq:stable} holds with approximation error $3\epsilon$, and the proof is complete. \qedhere
\end{proof}

As an immediate consequence, we obtain the following corollary.
\begin{corollary}
\label{corollary:sw-optimal-eq}
Consider any polymatrix game. For any strategy profile $\pi$ that maximizes social welfare, \emph{i.e.,}
$\pi \in \argmax_{\pi'\in \bigtimes_{i=1}^N \Delta^{\cA_i}} SW(\pi')$, there exists a payment plan $\bp$ such that $\rbr{\pi,\bp}$ is a SETE.
\end{corollary}

The claim follows because any maximizer of $SW$ is a $0$-stationary point of $SW$. Applying \Cref{lemma:social-welfare} with $\epsilon=0$ yields the result. Moreover, \Cref{corollary:sw-optimal-eq} provides an alternative route to establishing the existence of \Cref{def:SETE} in polymatrix games, independent of Nash equilibrium: since $\bigtimes_{i=1}^N \Delta^{\cA_i}$ is compact and $SW$ is continuous, $SW$ attains a maximizer by the Weierstrass theorem. This parallels the classic minimax-based existence proof of correlated equilibrium \citep{hart1989existence-phi-zero-sum}, which also suggests an efficient computation method (via solving the corresponding minimax problem). In the same spirit, our procedure in \Cref{eq:update-rule-potential-game} operationalizes this existence proof by explicitly searching for an (approximate) stationary point of the social welfare function. 

Finally, we present an example illustrating \Cref{corollary:sw-optimal-eq}: a socially optimal strategy profile admits a supporting payment plan $\bp$, and hence can be realized as a SETE.

\begin{example}
\label{example:PD-SETE}
In the Prisoner's Dilemma in \Cref{table:PD}, the unique NE is $(D,D)$ since $D$ strictly dominates $C$, yielding utility $0.2$ for both players. Now consider the deterministic profile $\pi_1(C)=\pi_2(C)=1$ together with the symmetric payment plan $p_{1\to 2}=p_{2\to 1}=0.4$. Under $(C,C)$, each player's net utility equals $0.6-0.4+0.4=0.6$, so both players strictly improve relative to $(D,D)$. Moreover, since $0.4 \ge 1.0-0.6$, no player has an incentive to deviate by \Cref{eq:stable}. Finally, $0.6-0.0 \ge 0.4$, so the promised transfers are individually rational by \Cref{eq:compatible}.
\end{example}

\subsection{Efficient Computation in Contrast to {\sf PPAD}-hardness of NE}

In \citet{chen2006settling-ppad,daskalakis2009complexity-PPAD}, it is shown that computing an $\epsilon$-approximate NE is {\tt PPAD}-hard, even for two-player normal-form games and constant $\epsilon$. In contrast, we show that in polymatrix games, which strictly generalize two-player normal-form games as the special case of a single edge, one can compute an \Cref{def:approximate-SETE} in polynomial time for any fixed $\epsilon$.

By \Cref{lemma:social-welfare}, it suffices to compute an approximate stationary point $\pi$ of the social welfare function $SW\colon\cA\to [0,N]$, where $SW(\ba)\coloneqq \sum_{j=1}^N \cU_j(\ba)$, and then construct a payment plan $\bp$ such that $(\pi,\bp)$ is an approximate \Cref{def:SETE} in the original game. We compute such an approximate stationary point by applying projected-gradient ascent to each player's strategy, using the social welfare as a common potential. We seek a stationary point rather than a maximizer because maximizing social welfare in polymatrix games is {\sf NP}-hard (see \Cref{lemma:polymatrix-optimal-sw-hard}).

Specifically, at each timestep $t\geq 1$, each player $i\in[N]$ updates $\overbar \pi_i^{(t)}$ using a payoff vector $\overbar\bu_i^{(t)}\in\RR^{\cA_i}$, whose coordinates are given by $\overbar u_i^{(t)}(a_i)$ for $a_i\in\cA_i$.
\begin{equation}
\begin{split}
    \overbar \pi_i^{(t+1)}=&\Proj{\Delta^{\cA_i}}{\overbar \pi_i^{(t)} + \eta\overbar\bu_i^{(t)}}\\
    \overbar u_i^{(t)}(a_i)=& \sum_{j=1}^N \EE_{\ba'\sim\overbar \pi^{(t)}}\sbr{\cU_j(a_i, \ba_{-i}')}\\
    \overbar p_{i\to j}^{(t)} =& \EE_{\ba\sim\overbar \pi^{(t)}}\sbr{\cU_i\rbr{\ba}} - \min_{\hat a_j\in\cA_j} \EE_{\ba\sim\overbar \pi^{(t)}}\sbr{\cU_i\rbr{\hat a_j, \ba_{-j}}},
\end{split}
\label{eq:update-rule-potential-game}
\end{equation}
Here, $\Proj{\Delta^{\cA_i}}{\cdot}$ denotes the Euclidean projection onto $\Delta^{\cA_i}$, and $\eta>0$ is the learning rate. We initialize $\overbar \pi_i^{(1)}$ to be the uniform distribution over $\cA_i$ for every $i\in[N]$.

Next, we establish the convergence rate of \Cref{eq:update-rule-potential-game} toward a \Cref{def:SETE}.
\begin{theorem}
\label{theorem:convergence-potential-game}
Consider \Cref{eq:update-rule-potential-game}. For any $T\geq 1$ and $\eta\leq \frac{1}{2NA}$, there exists $t^*\in[T]$ such that $(\overbar \pi^{(t^*)}, \overbar\bp^{(t^*)})$ satisfies \Cref{eq:compatible} and, for every $i\in[N]$,
\begin{align}
    \max_{\hat a_i\in\cA_i}\,
    \EE_{\ba\sim\overbar \pi^{(t^*)}}\sbr{\cU_i\rbr{\hat a_i,\ba_{-i}}}
    - \rbr{
        \EE_{\ba\sim\overbar \pi^{(t^*)}}\sbr{\cU_i\rbr{\ba}}
        + \sum_{\substack{j=1,\\ j\neq i}}^N \overbar p_{j\to i}^{(t^*)}
    }
    \leq 3\sqrt{\frac{2N}{\eta T}}.
\end{align}
\end{theorem}

\Cref{theorem:convergence-potential-game} implies that, for any fixed $\epsilon>0$, choosing
\begin{align*}
    T = \frac{18N}{\eta\epsilon^2}
\end{align*}
guarantees the existence of an iterate $t^*\in[T]$ such that $(\overbar \pi^{(t^*)}, \overbar\bp^{(t^*)})$ satisfies the $\epsilon$-approximate stability requirement \Cref{eq:approximate-stability} in \Cref{def:approximate-SETE}. Moreover, \Cref{eq:compatible} holds automatically by the definition of $\overbar\bp^{(t)}$. Consequently, with internal transfers enabled, one can compute an $\epsilon$-approximate product-form equilibrium in time polynomial in the game size for any fixed $\epsilon$, rather than facing \texttt{PPAD}-hardness.\footnote{This does not yield a polynomial-time algorithm in the bit complexity of the input, since the iteration complexity scales as $\cO\rbr{\frac{1}{\epsilon^2}}$ instead of $\text{Poly}\rbr{\log\frac{1}{\epsilon}}$.}

\begin{proof}
We begin by recalling the standard notion of smoothness. A differentiable function $f(\bx)\colon \cX\to\RR$ is called \emph{$2\beta$-smooth} if, for any $\bx,\bx'\in\cX$,
\begin{align*}
\abr{f\rbr{\bx'} - f\rbr{\bx} - \inner{\nabla f\rbr{\bx}}{\bx' - \bx}}
\leq
\beta \nbr{\bx' - \bx}^2.
\numberthis[$2\beta$-smoothness]{eq:smoothness}
\end{align*}

In \Cref{lemma:smoothness}, we show that $SW$ is $2\beta$-smooth with $\beta = NA$, where $A \coloneqq \max_{i\in[N]} \abr{\cA_i}$.

\begin{lemma}
\label{lemma:smoothness}
    The social welfare function $SW(\pi)\coloneqq \sum_{i=1}^N \EE_{\ba\sim\pi}\sbr{\cU_i(\ba)}$ is $2\beta$-smooth with respect to $\pi$. Formally, for any strategy profile $\pi, \pi'\in\bigtimes_{i=1}^N \Delta^{\cA_i}$,
    \begin{align*}
        &\nbr{\nabla SW(\pi') - \nabla SW(\pi)} \leq 2\beta\nbr{\pi' - \pi}\\
        &\abr{SW(\pi') - SW(\pi) - \inner{\nabla SW(\pi)}{\pi' - \pi}}\leq \beta \sum_{i=1}^N \nbr{\pi_i' - \pi_i}^2,
    \end{align*}
    where $\beta = NA$ and $A\coloneqq \max_{i\in[N]}|\cA_i|$ is the size of the largest action set.
\end{lemma}
The proof is postponed to \Cref{section:auxiliary-convergence-sw}. By the definition of smoothness, we have
\begin{equation}
\begin{split}
    SW(\overbar \pi^{(t+1)})\geq& SW(\overbar \pi^{(t)}) + \inner{\nabla SW(\overbar \pi^{(t)})}{\overbar \pi^{(t+1)} - \overbar \pi^{(t)}} - \beta\nbr{\overbar \pi^{(t+1)} - \overbar \pi^{(t)}}^2\\
    \overset{(i)}{=}& SW(\overbar \pi^{(t)}) + \sum_{i=1}^N \inner{\overbar \bu_i^{(t)}}{\overbar \pi_i^{(t+1)} - \overbar \pi_i^{(t)}} - \beta\sum_{i=1}^N \nbr{\overbar \pi_i^{(t+1)} - \overbar \pi_i^{(t)}}^2.
\end{split}
\label{eq:smoothness-sw}
\end{equation}
$(i)$ uses the fact that $\nabla_i SW(\overbar \pi^{(t)})=\overbar \bu_i^{(t)}$ (the formal proof can be found in \Cref{lemma:sw-grad-u}).

Next, we introduce the following lemma to bound $\inner{\overbar \bu_i^{(t)}}{\overbar \pi_i^{(t+1)} - \overbar \pi_i^{(t)}}$.

\begin{lemma}[Three-point Inequality]
\label{lemma:three-point}
    Let $\cX \subseteq \RR^n$ be a nonempty closed convex set. For any $\bx^{(0)},\bx^{(2)}\in\cX$, any $\bu\in \RR^n$, and any $\lambda>0$, define $\bx^{(1)}=\Proj{\cX}{\bx^{(0)} + \lambda \bu}$. Then,
    \begin{align*}
        \lambda\inner{\bu}{\bx^{(2)} - \bx^{(1)}}\leq \frac{1}{2}\rbr{\nbr{\bx^{(2)} - \bx^{(0)}}^2 - \nbr{\bx^{(2)} - \bx^{(1)}}^2 - \nbr{\bx^{(1)} - \bx^{(0)}}^2}.
    \end{align*}
\end{lemma}
The proof is postponed to \Cref{section:auxiliary-convergence-sw}. A more general version beyond Euclidean distance can be found in \citet[Lemma 10]{DBLP:conf/iclr/WeiLZL21-haipeng-normal-form-game} and \citet[Lemma 3.0.3]{liu2025solving-thesis}.

By taking 
\begin{align*}
    \cX=\Delta^{\cA_i},\quad \bx^{(1)}=\overbar \pi_i^{(t+1)},\quad \bx^{(2)} = \bx^{(0)} = \overbar \pi_i^{(t)},\quad \bu = \overbar \bu_i^{(t)},\quad \lambda = \eta
\end{align*}
in \Cref{lemma:three-point}, we have
\begin{align*}
    \eta\inner{\overbar \bu_i^{(t)}}{\overbar \pi_i^{(t+1)} - \overbar \pi_i^{(t)}}\geq& \frac{1}{2}\rbr{\nbr{\overbar \pi_i^{(t)} - \overbar \pi_i^{(t+1)}}^2 + \nbr{\overbar \pi_i^{(t+1)} - \overbar \pi_i^{(t)}}^2} = \nbr{\overbar \pi_i^{(t+1)} - \overbar \pi_i^{(t)}}^2.
\end{align*}
By substituting into \Cref{eq:smoothness-sw}, we have
\begin{align*}
    SW(\overbar \pi^{(t+1)})\geq& SW(\overbar \pi^{(t)}) + \frac{1}{\eta}\sum_{i=1}^N \nbr{\overbar \pi_i^{(t+1)} - \overbar \pi_i^{(t)}}^2 - \beta\sum_{i=1}^N \nbr{\overbar \pi_i^{(t+1)} - \overbar \pi_i^{(t)}}^2\\
    =&SW(\overbar \pi^{(t)}) + \rbr{\frac{1}{\eta} - \beta}\sum_{i=1}^N \nbr{\overbar \pi_i^{(t+1)} - \overbar \pi_i^{(t)}}^2.
\end{align*}
By telescoping and rearranging the terms, we have
\begin{align*}
    \rbr{\frac{1}{\eta} - \beta}\sum_{t=1}^{T} \sum_{i=1}^N \nbr{\overbar \pi_i^{(t+1)} - \overbar \pi_i^{(t)}}^2 \leq SW(\overbar \pi^{(T+1)}) - SW(\overbar \pi^{(1)})\overset{(i)}{\leq} N.
\end{align*}
$(i)$ is because the utility function $\cU_i$ of every player is bounded within $[0, 1]$. Hence, when $\eta\leq \frac{1}{2\beta}$, there must exist one timestep $t^*\in [T]$ such that
\begin{align*}
    \sum_{i=1}^N \nbr{\overbar \pi_i^{(t^*+1)} - \overbar \pi_i^{(t^*)}}^2 \leq \frac{2N}{T} \eta.
\end{align*}
By \Cref{lemma:variation-to-exploitability},
\begin{equation*}
    \max_{\substack{i\in [N],\\\hat a_i\in\cA_i}}~~ \sum_{j=1}^N \EE_{\ba\sim\overbar \pi^{(t^*)}} \sbr{\cU_j\rbr{\hat a_i, \ba_{-i}}} - SW(\overbar \pi^{(t^*)})\leq 3\sqrt{\frac{2N}{\eta T}}.
\end{equation*}
Next, we introduce \Cref{lemma:sw-to-SETE-gap} to bound the accuracy of \Cref{def:approximate-SETE}.
\begin{lemma}
\label{lemma:sw-to-SETE-gap}
    Consider \Cref{eq:update-rule-potential-game}. For any timestep $t\in [T]$, we have
    \begin{align*}
        &\max_{\substack{i\in [N],\\\hat a_i\in\cA_i}}~ \EE_{\ba\sim\overbar \pi^{(t)}} \sbr{\cU_i\rbr{\hat a_i, \ba_{-i}}} -\rbr{ \EE_{\ba\sim\overbar \pi^{(t)}} \sbr{\cU_i\rbr{\ba}} + \sum_{\substack{j=1,\\j\neq i}}^N \overbar p_{j\to i}^{(t)} } \\
        \leq& \max_{\substack{i\in [N],\\\hat a_i\in\cA_i}}~~ \sum_{j=1}^N \EE_{\ba\sim\overbar \pi^{(t)}} \sbr{\cU_j\rbr{\hat a_i, \ba_{-i}}} - SW(\overbar \pi^{(t)}).
    \end{align*}
\end{lemma}
The proof is postponed to \Cref{section:auxiliary-convergence-sw}. By \Cref{lemma:sw-to-SETE-gap} and the discussion above, we conclude the proof. \qedhere
\end{proof}

\section{Learning in Polymatrix Games}
\label{section:decentralized-learning}

\begin{algorithm}[h]
    \SetAlgoNoLine
    \caption{Learning Dynamics}
    \label{alg:learning}

    \KwIn{Learning rate $\eta$, exploration factor $\gamma\in(0,1]$, total number of timesteps $T$}

    \For{$t=1,2,\dots,T$}{
    
        \For{$i=1,2,\dots,N$}{
        
            Sample $a_i^{(t)}\sim \pi_i^{e,(t)}$ to explore, where
            \begin{align*}
                \pi_i^{e,(t)} = (1-\gamma) \pi_i^{(t)} + \frac{\gamma}{|\cA_i|}\one
            \end{align*}

            \For{$j\in \cN(i)$}{

                Compute player $i$'s payment to player $j$ at timestep $t$ for taking action $a_j^{(t)}$:
                \begin{minipage}{0.874\linewidth}
                \begin{align}
                    \tilde p_{i\to j}^{(t)} \gets \cU_{i, j}\rbr{a_i^{(t)}, a_j^{(t)}}  - \min_{\hat a_j\in\cA_j} \cU_{i, j}\rbr{a_i^{(t)}, \hat a_j} \label{eq:inst-payment}
                \end{align}
                \end{minipage}

                Send $\tilde p_{i\to j}^{(t)}$ to player $j$

            }
        }

        \For{$i=1,2,\dots,N$}{
            Compute the estimated utility vector $\tilde \bu_i^{(t)}$ as
            \begin{minipage}{0.912\linewidth}
            \begin{align}
                \tilde u_i^{(t)}(a_i)\gets\begin{cases}
                    \frac{1}{\pi_i^{e,(t)}(a_i^{(t)})}\sum_{j\in\cN(i)} \rbr{\cU_{i, j}\rbr{a_i^{(t)}, a_j^{(t)}} + \tilde p_{j\to i}^{(t)}} & a_i = a_i^{(t)}\\
                    0 & a_i \neq a_i^{(t)}
                \end{cases}
                \label{eq:importance-sampling}
            \end{align}
            \end{minipage}

            Update the strategy $\pi_i^{(t+1)}$ by
            \begin{minipage}{0.912\linewidth}
            \begin{align}
                \pi_i^{(t+1)} = \Proj{\Delta^{\cA_i}}{\pi_i^{(t)} + \eta\tilde\bu_i^{(t)}} \label{eq:PGA}
            \end{align}
            \end{minipage}
        }
    }
\end{algorithm}

In this section, we propose a decentralized learning procedure for \Cref{def:SETE}, as shown in \Cref{alg:learning}.
At each timestep $t$, each player $i$ draws an action $a_i^{(t)}$ from an exploration policy $\pi_i^{e,(t)}$ obtained by mixing the current strategy with the uniform distribution (Line 3). Given the realized joint action $\ba^{(t)}$, each player $i$ then computes an edge-wise transfer $\tilde p_{i\to j}^{(t)}$ for every neighbor $j\in\cN(i)$ (Line 5). The quantity in \Cref{eq:inst-payment} is exactly player $i$'s worst-case loss, holding $a_i^{(t)}$ fixed, from a deviation by $j$, and thus serves as a pessimistic estimate for $i$'s willingness to pay. 

The role of the transfer $\tilde p_{i\to j}^{(t)}$ is to make actions that benefit player $i$ more attractive to player $j$. Indeed, for fixed $a_i^{(t)}$, this transfer differs from $U_{i,j}(a_i^{(t)},a_j^{(t)})$ only by a constant that does not depend on $a_j^{(t)}$. Hence, from player $j$'s perspective, the incoming transfer from $i$ adds player $i$'s pairwise utility on the edge $(i,j)$ to player $j$'s own utility, up to an action-independent shift. With incoming transfers from all neighbors, player $j$'s utility estimate is therefore aligned with its contribution to social welfare, rather than with its private utility alone.

Finally, each player aggregates the observed pairwise utilities and the received transfers, uses importance sampling to construct an unbiased bandit estimate of the (augmented) payoff vector \Cref{eq:importance-sampling}, and performs a projected gradient step \Cref{eq:PGA}.

Then, we show that \Cref{alg:learning} enjoys best-iterate convergence in \Cref{theorem:convergence-sw}.

\begin{theorem}
\label{theorem:convergence-sw}
Consider \Cref{alg:learning}. For any $T\geq 1$, $\gamma\in (0,1]$, and $\eta\leq \frac{1}{2N A}$, define, for all $i\neq j$,
\begin{align}
    p_{i\to j}^{(t)}
    \;\coloneqq\;
    \EE_{\ba\sim\pi^{(t)}}\sbr{\cU_i\rbr{\ba}}
    \;-\;
    \min_{\hat a_j\in\cA_j}\EE_{\ba\sim\pi^{(t)}}\sbr{\cU_i\rbr{\hat a_j,\ba_{-j}}}.
\end{align}
Then there exists $t^*\in[T]$ such that $\rbr{\pi^{(t^*)},\bp^{(t^*)}}$ satisfies \Cref{eq:compatible} and, for every $i\in[N]$,
\begin{align}
    &\EE\sbr{\max_{\hat a_i\in\cA_i}\;
    \EE_{\ba\sim \pi^{(t^*)}}
    \sbr{\cU_i\rbr{\hat a_i,\ba_{-i}}}
    \;-\;
    \rbr{
        \EE_{\ba\sim \pi^{(t^*)}}\sbr{\cU_i\rbr{\ba}}
        + \sum_{\substack{j=1,\\ j\neq i}}^N p_{j\to i}^{(t^*)}
    }}\notag\\
    \;\le\;&
    18\sqrt{\frac{N}{4\eta T} + \frac{\eta}{2\gamma}N^4 A^3 + 4N^4 A^2\gamma},
\end{align}
where the expectation is taken over the randomness induced by sampling $\rbr{\ba^{(1)},\ba^{(2)},\dots,\ba^{(T)}}$.
\end{theorem}

By taking $\gamma=T^{-1/3}$ and $\eta=cT^{-2/3}$ for any constant $c\leq 1/(2NA)$, we obtain
\begin{align*}
    \EE\sbr{\max_{\hat a_i\in\cA_i}\;
    \EE_{\ba\sim \pi^{(t^*)}}
    \sbr{\cU_i\rbr{\hat a_i,\ba_{-i}}}
    \;-\;
    \rbr{
        \EE_{\ba\sim \pi^{(t^*)}}\sbr{\cU_i\rbr{\ba}}
        + \sum_{\substack{j=1,\\ j\neq i}}^N p_{j\to i}^{(t^*)}
    }}
    \;\le\;
    \cO\rbr{T^{-1/6}}.
\end{align*}

\begin{proof}
Motivated by \citet[Proof of Theorem 5.2]{panageas2023semi-bandit-congestion-game}, we use the Moreau envelope defined as
\begin{align*}
    M(\pi)\coloneqq \max_{\pi'\in\bigtimes_{i=1}^N \Delta^{\cA_i}} \rbr{SW(\pi') - 2\beta\nbr{\pi' - \pi}^2},
\end{align*}
where $\beta$ is the smoothness constant of social welfare $SW$ defined in \Cref{lemma:smoothness}. By letting
\begin{align*}
    \mu^{(t+1)}\coloneqq \argmax_{\mu\in\bigtimes_{i=1}^N \Delta^{\cA_i}} \rbr{SW(\mu) - 2\beta\nbr{\mu - \pi^{(t)}}^2},
\end{align*}
we have
\begin{align*}
    M(\pi^{(t+1)})\overset{(i)}{\geq}& SW(\mu^{(t+1)}) - 2\beta \nbr{\mu^{(t+1)} - \pi^{(t+1)}}^2\\
    =& SW(\mu^{(t+1)}) - 2\beta \nbr{\mu^{(t+1)} - \Proj{\bigtimes_{i=1}^N \Delta^{\cA_i}}{\pi^{(t)} + \eta\tilde\bu^{(t)}}}^2\\
    \overset{(ii)}{\geq}& SW(\mu^{(t+1)}) - 2\beta \nbr{\mu^{(t+1)} - \pi^{(t)} - \eta\tilde\bu^{(t)}}^2\\
    =&SW(\mu^{(t+1)}) - 2\beta \nbr{\mu^{(t+1)} - \pi^{(t)}}^2 + 4\beta\eta\inner{\mu^{(t+1)} - \pi^{(t)}}{\tilde\bu^{(t)}} - 2\beta\eta^2\nbr{\tilde\bu^{(t)}}^2\\
    \overset{(iii)}{=}& M(\pi^{(t)}) + 4\beta\eta\inner{\mu^{(t+1)} - \pi^{(t)}}{\tilde\bu^{(t)}} - 2\beta\eta^2\nbr{\tilde\bu^{(t)}}^2.
\end{align*}
$(i)$ is by the definition of $M$. $(ii)$ is by the property of projection. $(iii)$ is by the definition of $\mu^{(t+1)}$. By taking the expectation on both sides, we have
\begin{align*}
    \EE\sbr{M(\pi^{(t+1)})}\geq \EE\sbr{M(\pi^{(t)})} + 4\beta\eta\EE\sbr{\inner{\mu^{(t+1)} - \pi^{(t)}}{\tilde\bu^{(t)}}} - 2\beta\eta^2\EE\sbr{\nbr{\tilde\bu^{(t)}}^2}.
\end{align*}
By \Cref{lemma:bounded-variance} in the following, we can bound $2\beta\eta^2\EE\sbr{\nbr{\tilde\bu^{(t)}}^2}$ above.
\begin{lemma}
\label{lemma:bounded-variance}
For any timestep $t\in [T]$, we have
\begin{align*}
    \EE\sbr{\nbr{\tilde\bu^{(t)}}^2} \leq \frac{\sigma_u^2}{\gamma},
\end{align*}
where $\sigma_u\coloneqq A N^{3/2}$ and $A=\max_{i\in [N]}|\cA_i|$.
\end{lemma}
The proof is postponed to \Cref{section:bandit-auxiliary}. Hence,
\begin{align}
    \EE\sbr{M(\pi^{(t+1)})}\geq& \EE\sbr{M(\pi^{(t)})} + 4\beta\eta\EE\sbr{\inner{\mu^{(t+1)} - \pi^{(t)}}{\tilde\bu^{(t)}}} - \frac{2\beta\eta^2}{\gamma}\sigma_u^2\notag\\
    =&\EE\sbr{M(\pi^{(t)})} + 4\beta\eta\EE\sbr{\inner{\mu^{(t+1)} - \pi^{(t)}}{\EE\sbr{\tilde\bu^{(t)}\given \pi^{(t)}}}} - \frac{2\beta\eta^2}{\gamma}\sigma_u^2.\label{eq:bandit-2}
\end{align}
Then, as shown in \Cref{lemma:unbiased-estimator}, $\tilde\bu_i^{(t)}$ is an unbiased estimator of $\nabla_i SW(\pi^{e, (t)})-c_i^{(t)} \one$, where $c_i^{(t)}\in \RR$ is a scalar and $\one\in \RR^{\cA_i}$ denotes the all-ones vector indexed by $\cA_i$.

\begin{lemma}
\label{lemma:unbiased-estimator}
For any timestep $t\in [T]$, any player $i\in [N]$, and any action $a_i\in\cA_i$, we have
\begin{align}
    \EE\sbr{\tilde u_i^{(t)}(a_i)\given \pi^{(t)}} =& \frac{\partial SW(\pi^{e,(t)})}{\partial \pi_i^{e,(t)}(a_i)} - c_i^{(t)},\label{eq:bandit-1}
\end{align}
with
\begin{align*}
    c_i^{(t)} = \sum_{j\in\cN(i)} \sum_{a_j\in\cA_j} \rbr{\min_{\hat a_i\in\cA_i} \cU_{j, i}\rbr{a_j, \hat a_i}} \pi_j^{e,(t)}(a_j).
\end{align*}
\end{lemma}
The proof is postponed to \Cref{section:bandit-auxiliary}. Therefore,
\begin{align*}
    &\EE\sbr{\inner{\mu^{(t+1)} - \pi^{(t)}}{\EE\sbr{\tilde\bu^{(t)}\given \pi^{(t)}}}}\\
    =& \sum_{i=1}^N \EE\sbr{\inner{\mu_i^{(t+1)} - \pi_i^{(t)}}{\EE\sbr{\tilde\bu_i^{(t)}\given \pi^{(t)}}}}\\
    \overset{(i)}{=}& \sum_{i=1}^N \EE\sbr{\inner{\mu_i^{(t+1)} - \pi_i^{(t)}}{\nabla_i SW(\pi^{e,(t)})}}\\
    =& \EE\sbr{\inner{\mu^{(t+1)} - \pi^{(t)}}{\nabla SW(\pi^{(t)})}} + \EE\sbr{\inner{\mu^{(t+1)} - \pi^{(t)}}{\nabla SW(\pi^{e,(t)}) - \nabla SW(\pi^{(t)})}}.
\end{align*}
$(i)$ is because $\inner{\mu_i^{(t+1)} - \pi_i^{(t)}}{c_i^{(t)}\one} = c_i^{(t)} - c_i^{(t)} = 0$. Moreover,
\begin{align*}
    \abr{\inner{\mu^{(t+1)} - \pi^{(t)}}{\nabla SW(\pi^{e,(t)}) - \nabla SW(\pi^{(t)})}}\overset{(i)}{\leq}& \nbr{\mu^{(t+1)} - \pi^{(t)}}\cdot\nbr{\nabla SW(\pi^{e,(t)}) - \nabla SW(\pi^{(t)})}\\
    \overset{(ii)}{\leq}& 2\sqrt{2} \beta N\nbr{\pi^{e,(t)} - \pi^{(t)}}\\
    \overset{(iii)}{\leq}& 2\sqrt{2} \beta N \sum_{i=1}^N \gamma\nbr{\frac{\one}{|\cA_i|} - \pi_i^{(t)}}\\
    \leq& 4 \beta \gamma N^2.
\end{align*}
$(i)$ is by \holder. $(ii)$ is by \Cref{lemma:smoothness}. $(iii)$ is by the definition of $\pi_i^{e,(t)}$ in \Cref{alg:learning}.

Finally, by substituting back to \Cref{eq:bandit-2}, we have
\begin{align}
    \EE\sbr{M(\pi^{(t+1)})}\geq& \EE\sbr{M(\pi^{(t)})} + 4\beta\eta\EE\sbr{\inner{\mu^{(t+1)} - \pi^{(t)}}{\nabla SW(\pi^{(t)})}} - 16\beta^2\eta\gamma N^2 - \frac{2\beta\eta^2}{\gamma}\sigma_u^2.\label{eq:bandit-3}
\end{align}

By \Cref{lemma:smoothness}, we have
\begin{align*}
    &\inner{\mu^{(t+1)} - \pi^{(t)}}{\nabla SW(\pi^{(t)})}\\
    \geq& SW(\mu^{(t+1)}) - SW(\pi^{(t)}) - \beta \nbr{\mu^{(t+1)} - \pi^{(t)}}^2\\
    =& \rbr{SW(\mu^{(t+1)}) - 2\beta \nbr{\mu^{(t+1)} - \pi^{(t)}}^2} -\rbr{SW(\pi^{(t)}) - 2\beta \nbr{\pi^{(t)} - \pi^{(t)}}^2} + \beta \nbr{\mu^{(t+1)} - \pi^{(t)}}^2\\
    \overset{(i)}{\geq}& \beta \nbr{\mu^{(t+1)} - \pi^{(t)}}^2.
\end{align*}
$(i)$ is by the definition of $\mu^{(t+1)}$. Therefore, by substituting back to \Cref{eq:bandit-3} and rearranging terms, we have
\begin{align*}
    4\beta^2\eta\EE\sbr{\nbr{\mu^{(t+1)} - \pi^{(t)}}^2}\leq \EE\sbr{M(\pi^{(t+1)})} - \EE\sbr{M(\pi^{(t)})} + 16\beta^2\eta\gamma N^2 + \frac{2\beta\eta^2}{\gamma}\sigma_u^2.
\end{align*}
By telescoping,
\begin{align*}
    \sum_{t=1}^T \EE\sbr{\nbr{\mu^{(t+1)} - \pi^{(t)}}^2}\leq& \frac{1}{4\beta^2\eta}\rbr{\EE\sbr{M(\pi^{(T+1)})} - \EE\sbr{M(\pi^{(1)})}} + \frac{\eta}{2\beta\gamma}\sigma_u^2 T + 4\gamma N^2 T\\
    \overset{(i)}{\leq}& \frac{N}{4\beta^2\eta} + \frac{\eta}{2\beta\gamma}\sigma_u^2 T + 4\gamma N^2 T.
\end{align*}
$(i)$ is because for any $\pi\in\bigtimes_{i=1}^N \Delta^{\cA_i}$,
\begin{align*}
    0\leq SW(\pi) - 2\beta\nbr{\pi-\pi}^2\leq M(\pi) \leq \max_{\pi'\in\bigtimes_{i=1}^N \Delta^{\cA_i}} SW(\pi')\leq N.
\end{align*}

By the following lemma, we can bound the gap to NE of the potential game.
\begin{lemma}
\label{lemma:bandit-distance-NE}
For any strategy profile $\pi\in \bigtimes_{i=1}^N \Delta^{\cA_i}$, let 
\begin{align*}
    \mu=\argmax_{\mu'\in \bigtimes_{i=1}^N \Delta^{\cA_i}} \rbr{SW(\mu') - 2\beta\nbr{\mu' - \pi}^2}.
\end{align*}
Then,
\begin{align*}
    \max_{\substack{i\in [N]\\\hat\pi_i\in\Delta^{\cA_i}}}\inner{\nabla_i SW(\pi)}{\hat \pi_i - \pi_i} \leq 18\beta \nbr{\mu - \pi}.
\end{align*}
\end{lemma}
The proof is postponed to \Cref{section:bandit-auxiliary}. Therefore,
\begin{align*}
    \frac{1}{T}\sum_{t=1}^T \EE\sbr{\max_{\substack{i\in [N]\\\hat\pi_i\in\Delta^{\cA_i}}}\inner{\nabla_i SW(\pi^{(t)})}{\hat \pi_i - \pi_i^{(t)}}} \leq& 18\beta \rbr{\frac{1}{T}\sum_{t=1}^T \EE\sbr{\nbr{\mu^{(t+1)} - \pi^{(t)}}}}\\
    \overset{(i)}{\leq}& 18\beta \sqrt{\frac{1}{T}\sum_{t=1}^T \rbr{\EE\sbr{\nbr{\mu^{(t+1)} - \pi^{(t)}}}}^2}\\
    \overset{(ii)}{\leq}& 18\beta \sqrt{\frac{1}{T}\sum_{t=1}^T \EE\sbr{\nbr{\mu^{(t+1)} - \pi^{(t)}}^2}}\\
    \leq& 18\sqrt{\frac{N}{4\eta T} + \frac{\eta\beta}{2\gamma}\sigma_u^2 + 4\beta^2\gamma N^2}\\
    \leq& 18\sqrt{\frac{N}{4\eta T} + \frac{\eta}{2\gamma}N^4 A^3 + 4N^4 A^2\gamma},
\end{align*}
where the last inequality uses the smoothness constant $\beta=NA$, the variance constant $\sigma_u^2=N^3A^2$, and $N\geq 1$. $(i)$ and $(ii)$ are both by \jensen. Hence, there must exist $t^*\in [T]$ such that
\begin{align*}
    \EE\sbr{\max_{\substack{i\in [N]\\\hat\pi_i\in\Delta^{\cA_i}}}\inner{\nabla_i SW(\pi^{(t^*)})}{\hat \pi_i - \pi_i^{(t^*)}}}\leq& \frac{1}{T}\sum_{t=1}^T \EE\sbr{\max_{\substack{i\in [N]\\\hat\pi_i\in\Delta^{\cA_i}}}\inner{\nabla_i SW(\pi^{(t)})}{\hat \pi_i - \pi_i^{(t)}}}\\
    \leq& 18\sqrt{\frac{N}{4\eta T} + \frac{\eta}{2\gamma}N^4 A^3 + 4N^4 A^2\gamma}.
    \end{align*}
    The SETE exploitability bound in the theorem follows from \Cref{lemma:sw-to-SETE-gap}, applied to the payment plan $\bp^{(t^*)}$. \qedhere
\end{proof}

\section{Mediated Self-Enforcing Transfer Equilibrium (M-SETE)}
\label{section:mediated-SETE}

\begin{figure}[t]
\centering
\scalebox{0.9}{
\begin{tikzpicture}[
  level distance=1.5cm,
  level 1/.style={sibling distance=8cm},
  level 2/.style={sibling distance=4cm},
  level 3/.style={sibling distance=2cm},
  edge from parent/.style={draw,thick}
]
\node (root) {$\btau = (\tau_1, \tau_2)$}
  child { node (A) {$\pi=\rbr{\pi_1,\pi_2}$}
    child { node (B) {$\EE_{\ba\sim \pi}\sbr{\cU_1(\ba)} - \tau_1$}
            edge from parent node[midway, yshift=4pt, left=4pt] {Follow $\pi_1$}
        }
    child { node (C) {\color{red} Illegal (Binding Offer)}
            edge from parent node[midway, yshift=4pt, right=4pt] {Deviate to $\hat a_1\in\cA_1$}
            node[midway, below=-14pt, red, scale=1.5]{
                \tikz{
                  \draw[red, line width=1pt] (0,0) circle (0.14cm);
                  \draw[red, line width=1pt] (-0.10cm,0.10cm) -- (0.10cm,-0.10cm);
                }
              }
        }
    edge from parent node[midway, yshift=4pt, left=4pt] {Accept the Mechanism}
  }
  child { node (D) {$\min_{\hat \pi_2\in\Delta^{\cA_2}} \max_{\hat \pi_1\in \Delta^{\cA_1}} \EE_{\hat \ba\sim \hat\pi} \sbr{\cU_1(\hat \ba)}$}
    edge from parent node[midway, yshift=4pt, right=4pt] {Decline the Mechanism}
  };

\draw[green,thick] (root) -- (A);
\draw[green,thick] (A) -- (B);
\draw[red,thick] (A) -- (C);
\draw[red,thick] (root) -- (D);
\end{tikzpicture}
}
\caption{Illustration of the augmented two-stage game, taking player 1 as the player of interest. The green path forms a Nash equilibrium of the sequential game induced by binding offers: accepting the mechanism binds both the payment and the recommended strategy.}
\label{fig:M-SETE-augmented-game}
\end{figure}
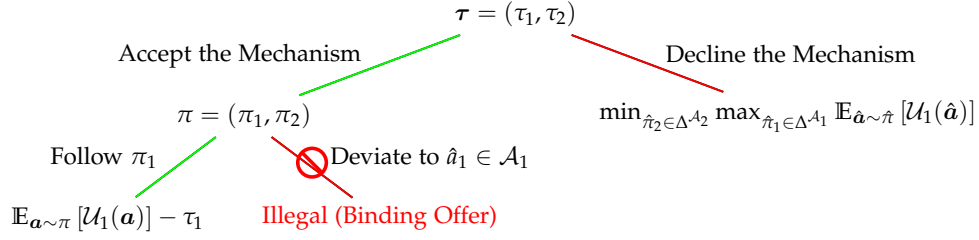

\Cref{def:SETE} only guarantees NE of the agent normal form of the augmented game, which is weaker than NE of the original augmented game. In particular, a player $i \in [N]$ may still profit by deviating across both stages: first withdrawing all promised payments and then deviating from $\pi_i$. As shown in \Cref{lemma:sw-optimal-not-NE}, there exists a two-player game in which the socially optimal strategy profile is not a NE of the augmented game. Therefore, to guarantee NE in the original augmented game, one would need a stronger mechanism.

We now introduce a mediator into the mechanism. Through legal enforcement, the mediator makes the payment schedule and the recommended strategy binding offers to each player. Hence, each player can either accept the entire offer or decline the mechanism altogether. In particular, a player $i \in [N]$ can no longer use payments to steer the other players toward $\pi_{-i}$ while simultaneously deviating from $\pi_i$. Because the offer is binding, the peer-to-peer transfers can be aggregated into a lump-sum payment $\tau_i \in \RR$, where $\tau_i<0$ denotes a subsidy specified by the mediator. The budget-balance condition below ensures that these subsidies are funded by other players' payments rather than by an external injection of money. Define
\begin{align*}
    \tau_i=\sum_{\substack{j=1,\\j\neq i}}^N \rbr{p_{i\to j} - p_{j\to i}}.
\end{align*}
The resulting mechanism, illustrated in \Cref{fig:M-SETE-augmented-game}, proceeds as follows:
\begin{enumerate}
    \item The mediator announces $(\pi, \btau)$ to all players.
    \item Each player decides whether to join the mechanism.
    \item If player $i \in [N]$ joins, then player $i$ must pay $\tau_i$ and follow the action recommendation $a_i \in \cA_i$ issued by the mediator, where $a_i \sim \pi_i$.
\end{enumerate}

Then, the constraints become
\begin{align}
    & \sum_{i=1}^N \tau_i = 0\label{eq:M-SETE-constraint-1}\\
    \forall i\in [N],\quad & \EE_{\ba\sim \pi} \sbr{\cU_i(\ba)} - \tau_i \geq \min_{\hat \pi_{-i}\in\bigtimes_{j\in [N]\setminus\cbr{i}} \Delta^{\cA_j}} \max_{\hat \pi_i\in \Delta^{\cA_i}} \EE_{\hat \ba\sim \hat\pi} \sbr{\cU_i(\hat \ba)}.\label{eq:M-SETE-constraint-2}
\end{align}
\Cref{eq:M-SETE-constraint-1} enforces budget balance. \Cref{eq:M-SETE-constraint-2} ensures that no player $i \in [N]$ has an incentive to quit the mechanism: if player $i$ joins, then their utility is $\EE_{\ba \sim \pi} \sbr{\cU_i(\ba)} - \tau_i$, whereas if player $i$ declines the mechanism, their utility is
\begin{align*}
    v_i\coloneqq \min_{\hat \pi_{-i}\in\bigtimes_{j\in [N]\setminus\cbr{i}} \Delta^{\cA_j}} \max_{\hat \pi_i\in \Delta^{\cA_i}} \EE_{\hat \ba\sim \hat\pi} \sbr{\cU_i(\hat \ba)}.
\end{align*}
The minimax value is a pessimistic estimate of player $i$'s utility after leaving the mechanism. Indeed, without knowledge of the other players' utility functions, player $i$ cannot predict how the others will behave once $i$ withdraws from the mechanism. Equivalently, the continuation strategies of the other players may be viewed as a non-credible threat against the deviating player. In particular, in the subgame induced by player $i$ leaving the mechanism, one may take
\begin{align*}
    \pi_{-i}\in \argmin_{\hat \pi_{-i}\in\bigtimes_{j\in [N]\setminus\cbr{i}} \Delta^{\cA_j}} \max_{\hat \pi_i\in \Delta^{\cA_i}} \EE_{\hat \ba\sim \hat\pi} \sbr{\cU_i(\hat \ba)}.
\end{align*}
We are now ready to define the resulting equilibrium concept.

\begin{definition}[Mediated Self-Enforcing Transfer Equilibrium (M-SETE)]
    In any game, a Mediated Self-Enforcing Transfer Equilibrium (M-SETE) is a strategy-payment pair $(\pi, \btau)$ that satisfies \Cref{eq:M-SETE-constraint-1,eq:M-SETE-constraint-2} simultaneously, where $\pi = (\pi_i)_{i=1}^N$ is a strategy profile and $\btau = (\tau_i)_{i=1}^N$ is a payment plan. Here, $\pi_i$ is the strategy that player $i$ must follow upon accepting the mechanism, and $\tau_i$ is the payment, or subsidy if negative, specified by the mediator that player $i$ must pay upon accepting the mechanism.
\end{definition}

With such a mediator, the result extends to general games and yields an NE of the augmented game induced by the binding-offer mechanism, as stated in the following theorem.

\begin{theorem}
\label{theorem:M-SETE}
    Any M-SETE $(\pi, \btau)$ is a Nash equilibrium of the augmented game induced by the binding-offer mechanism, not merely of its agent normal form. Moreover, for any strategy profile $\pi\in\bigtimes_{i=1}^N \Delta^{\cA_i}$ satisfying
    \begin{align}
        \sum_{i=1}^N \EE_{\ba\sim\pi}\sbr{\cU_i(\ba)} \geq \sum_{i=1}^N v_i,\label{eq:sw-higher-than-minimax}
    \end{align}
    there exists a payment plan $\btau \in \RR^N$ such that $(\pi, \btau)$ is an M-SETE.
\end{theorem}

\begin{proof}
    If a player accepts the mechanism, the binding offer prevents that player from separately deviating from the prescribed payment or from the prescribed strategy within the mechanism. The only remaining unilateral deviation is to decline the offer. \Cref{eq:M-SETE-constraint-2} guarantees that no player has an incentive to decline the mechanism, under the continuation value imposed on a deviating player. Hence any M-SETE is a Nash equilibrium of the augmented game.

    For any strategy profile $\pi\in\bigtimes_{i=1}^N \Delta^{\cA_i}$ satisfying \Cref{eq:sw-higher-than-minimax}, let
    \begin{align*}
        \tau_i = \EE_{\ba\sim\pi}\sbr{\cU_i(\ba)} - v_i - \frac{1}{N}\rbr{\sum_{j=1}^N \EE_{\ba\sim\pi}\sbr{\cU_j(\ba)} - \sum_{j=1}^N v_j}.
    \end{align*}
    Then, $(\pi, \btau)$ satisfies \Cref{eq:M-SETE-constraint-2} since $\pi$ satisfies \Cref{eq:sw-higher-than-minimax}. Moreover, $\sum_{i=1}^N \tau_i = 0$ by definition, so the budget-balance condition \Cref{eq:M-SETE-constraint-1} also holds. \qedhere
\end{proof}

Furthermore, M-SETE always exists in any finite game, as stated in the following corollary.
\begin{corollary}
\label{corollary:NE-M-SETE}
    For any strategy profile $\pi \in \bigtimes_{i=1}^N \Delta^{\cA_i}$ whose social welfare is no smaller than that of the worst Nash equilibrium, there exists a payment plan $\btau \in \RR^N$ such that $(\pi, \btau)$ is an M-SETE. In particular, M-SETE always exists and any socially optimal strategy profile can be supported as an M-SETE in any finite game.
\end{corollary}

\begin{proof}
    For any Nash equilibrium $\pi^{\rm NE}$ and any player $i\in [N]$, we have
    \begin{align*}
        \EE_{\ba\sim \pi^{\rm NE}}\sbr{\cU_i(\ba)} =& \max_{\hat\pi_i\in\Delta^{\cA_i}} \EE_{\hat a_i\sim \hat\pi_i, \ba\sim \pi^{\rm NE}}\sbr{\cU_i(\hat a_i, \ba_{-i})}\\
        \geq& \min_{\hat \pi_{-i}\in\bigtimes_{j\in [N]\setminus\cbr{i}} \Delta^{\cA_j}} \max_{\hat \pi_i\in \Delta^{\cA_i}} \EE_{\hat \ba\sim \hat\pi} \sbr{\cU_i(\hat \ba)}\\
        =& v_i.
    \end{align*}
    Therefore, for any strategy profile $\pi$ with social welfare no smaller than that of the worst NE, we must have
    \begin{align*}
        \sum_{i=1}^N \EE_{\ba\sim \pi}\sbr{\cU_i(\ba)} \geq \sum_{i=1}^N v_i.
    \end{align*}
    By \Cref{theorem:M-SETE}, there must exist a payment plan $\btau\in \RR^N$ such that $(\pi, \btau)$ is an M-SETE.

    Finally, since every finite game admits at least one NE \citep{nash1950equilibrium-nash-def}, M-SETE always exists. Any socially optimal strategy profile can be supported as an M-SETE, since its social welfare is necessarily no smaller than that of the worst NE. \qedhere
\end{proof}

However, computing $v_i$ is {\sf NP}-hard \citep{borgs2010myth-folk-theorem}, so it is generally intractable as the continuation-value estimate faced by a player who contemplates deviating. In particular, a deviating player may be unable to compute this pessimistic benchmark exactly. We therefore introduce the following tractable proxy:
\begin{align}
    \tilde v_i\coloneqq \min_{\hat \pi_{-i}\in {\color{red}\Delta^{\cA_{-i}}}} \max_{\hat \pi_i\in \Delta^{\cA_i}} \EE_{\hat \ba\sim \hat\pi} \sbr{\cU_i(\hat \ba)}.
\end{align}
Here, $\hat\pi_{-i}\in \Delta^{\cA_{-i}}$ is a distribution over the joint action space of the other players, and therefore allows their actions to be correlated through the mediator. Since $\bigtimes_{j\in [N]\setminus\cbr{i}} \Delta^{\cA_j} \subseteq \Delta^{\cA_{-i}}$, it follows immediately that $v_i \geq \tilde v_i$ for every player $i\in [N]$. Accordingly, \Cref{eq:M-SETE-constraint-2} can be relaxed to
\begin{align}
    \forall i\in [N],\quad & \EE_{\ba\sim \pi} \sbr{\cU_i(\ba)} - \tau_i \geq \tilde v_i. \label{eq:CM-SETE-constraint-3}
\end{align}

This leads to the following equilibrium notion, the Correlated M-SETE (CM-SETE).
\begin{definition}[Correlated Mediated Self-Enforcing Transfer Equilibrium (CM-SETE)]
    In any game, a Correlated Mediated Self-Enforcing Transfer Equilibrium (CM-SETE) is a strategy-payment pair $(\pi, \btau)$ that satisfies \Cref{eq:M-SETE-constraint-1,eq:CM-SETE-constraint-3}, where $\pi=(\pi_i)_{i=1}^N$ is a strategy profile and $\btau=(\tau_i)_{i=1}^N$ is a payment plan. Here, $\pi_i$ is the strategy prescribed to player $i$ upon accepting the mechanism, and $\tau_i$ is the payment, or subsidy, if negative, specified by the mediator that player $i$ must pay upon accepting the mechanism.
\end{definition}

\begin{remark}
    Although CM-SETE allows correlation in the continuation following a player's deviation, no correlation is required along the equilibrium path. Thus, during implementation, players do not need the mediator to actively correlate their actions. The correlated continuation serves only as an off-path, non-credible threat.
\end{remark}

Additionally, \Cref{theorem:M-SETE} and \Cref{corollary:NE-M-SETE} admit direct analogues for CM-SETE. We state the result formally below.
\begin{theorem}
    \label{theorem:CM-SETE}
    Any CM-SETE $(\pi,\btau)$ is a Nash equilibrium of the augmented game when correlated strategies are allowed in the continuation following a deviation. Moreover, for any strategy profile $\pi\in\bigtimes_{i=1}^N \Delta^{\cA_i}$ satisfying
    \begin{align}
        \sum_{i=1}^N \EE_{\ba\sim\pi}\sbr{\cU_i(\ba)} \geq \sum_{i=1}^N \tilde v_i,\label{eq:sw-higher-than-correlated-minimax}
    \end{align}
    there exists a payment plan $\btau\in\RR^N$ such that $(\pi,\btau)$ is a CM-SETE. Consequently, the analogue of \Cref{corollary:NE-M-SETE} also holds for CM-SETE.
\end{theorem}
The proof follows directly from that of \Cref{theorem:M-SETE} and \Cref{corollary:NE-M-SETE}, with $\tilde v_i$ replacing $v_i$ as the continuation value after a deviation. The analogue of \Cref{corollary:NE-M-SETE} uses the fact that $v_i\geq \tilde v_i$. Moreover, by the minimax theorem, for every player $i\in [N]$,
\begin{align*}
    \min_{\hat \pi_{-i}\in \Delta^{\cA_{-i}}} \max_{\hat \pi_i\in \Delta^{\cA_i}} \EE_{\hat \ba\sim \hat\pi} \sbr{\cU_i(\hat \ba)} = \max_{\hat \pi_i\in \Delta^{\cA_i}} \min_{\hat \pi_{-i}\in \Delta^{\cA_{-i}}} \EE_{\hat \ba\sim \hat\pi} \sbr{\cU_i(\hat \ba)}.
\end{align*}
Therefore, if for every player $i\in [N]$ we choose
\begin{align*}
    \pi_i \in \argmax_{\hat \pi_i\in \Delta^{\cA_i}} \min_{\hat \pi_{-i}\in \Delta^{\cA_{-i}}} \EE_{\hat \ba\sim \hat\pi} \sbr{\cU_i(\hat \ba)},
\end{align*}
then \Cref{eq:sw-higher-than-correlated-minimax} is satisfied by the resulting strategy profile $\pi$. In fact, $(\pi,\zero)$ is a CM-SETE, since for every player $i\in [N]$,
\begin{align*}
    \EE_{\ba\sim\pi}\sbr{\cU_i(\ba)}\geq \max_{\hat \pi_i\in \Delta^{\cA_i}} \min_{\hat \pi_{-i}\in \Delta^{\cA_{-i}}} \EE_{\hat \ba\sim \hat\pi} \sbr{\cU_i(\hat \ba)} = \min_{\hat \pi_{-i}\in \Delta^{\cA_{-i}}} \max_{\hat \pi_i\in \Delta^{\cA_i}} \EE_{\hat \ba\sim \hat\pi} \sbr{\cU_i(\hat \ba)}=\tilde v_i.
\end{align*}

Therefore, CM-SETE can be computed by solving the following feasibility linear program for each player $i\in [N]$:
\begin{align*}
    \text{find}\quad & \pi_i, ~\mu_{-i}, ~\tilde v_i \\
    \text{s.t.}\quad& \forall \hat a_i\in\cA_i,\quad \tilde v_i\geq \EE_{\ba_{-i}\sim \mu_{-i}}\sbr{\cU_i(\hat a_i, \ba_{-i})} \\
    & \forall \hat \ba_{-i}\in \cA_{-i},\quad \tilde v_i\leq \EE_{a_i\sim \pi_i}\sbr{\cU_i(a_i, \hat \ba_{-i})} \\
    &\sum_{a_i\in \cA_i} \pi_i(a_i)=1\\
    &\sum_{\ba_{-i}\in \cA_{-i}} \mu_{-i}(\ba_{-i})=1\\
    &\forall a_i\in\cA_i,\quad \pi_i(a_i)\geq 0\\
    &\forall \ba_{-i}\in\cA_{-i},\quad \mu_{-i}(\ba_{-i})\geq 0.
\end{align*}
Letting $\pi=(\pi_i)_{i=1}^N$ denote the resulting strategy profile, we obtain that $(\pi,\zero)$ is a CM-SETE.

\begin{remark}
    Unlike the efficient-computation results for SETE in \Cref{section:polymatrix}, the concepts of M-SETE and CM-SETE discussed in this section are not restricted to polymatrix games. Instead, all results in this section apply to arbitrary finite games.
\end{remark}

\section{Experiments}
\label{section:experiment}

In this section, we compare the learning dynamics of two classical algorithms, projected gradient ascent (PGA) and EXP3 \citep{hazan2016introduction}, when run \emph{with} and \emph{without} payments. Note that PGA with payments coincides with \Cref{alg:learning}. Across random polymatrix games with different underlying graph topologies (see \Cref{table:last_iterate_exploitability,table:last_iterate_social_welfare}), we observe that allowing internal transfers increases social welfare and accelerates convergence toward equilibrium. 

The payment computation adds little overhead. For each edge $(i,j)$ and each action $a_i\in\cA_i$, the quantity $\min_{\hat a_j\in\cA_j}\cU_{i,j}\rbr{a_i,\hat a_j}$ can be precomputed before the learning procedure begins. After this preprocessing, computing each transfer $\tilde p_{i\to j}^{(t)}$ requires only a table lookup and one subtraction, and incorporating received transfers requires only the corresponding additions. Thus, relative to the payoff aggregation over neighbors already performed by the learning procedure, payments add only a constant number of basic arithmetic operations per edge at each timestep. In our experiments, for instance, when $N=256$ and $|\cA_i|=5$ for all players $i\in[N]$, the running time increases from $145.42$ seconds without payments to $153.33$ seconds with payments, an overhead of about $5.44\%$.

As shown in \Cref{fig:N-64-sparse}, introducing payments leads both PGA and EXP3 to achieve higher social welfare. We also observe a lower exploitability, defined for a strategy-payment pair $(\pi, \bp)$ as
\begin{align}
     \max_{\substack{i\in [N]\\\hat a_i\in\cA_i}} \rbr{\EE_{\ba\sim\pi} \sbr{\cU_i(\hat a_i, \ba_{-i}) - \cU_i(\ba)} - \sum_{\substack{j=1\\j\neq i}}^N p_{j\to i}}.\numberthis[Exploitability]{eq:exploitability}
\end{align}
When payments are not allowed, we set $\bp=\zero$.
Intuitively, because transfers are received only when players adhere to the prescribed strategies, a deviation forfeits incoming payments. This raises the opportunity cost of deviating and shrinks the maximum unilateral gain, thereby reducing exploitability.
More details about the experiments can be found in \Cref{section:experiment-details}.

\begin{figure}
    \centering
    \includegraphics[width=0.8\linewidth]{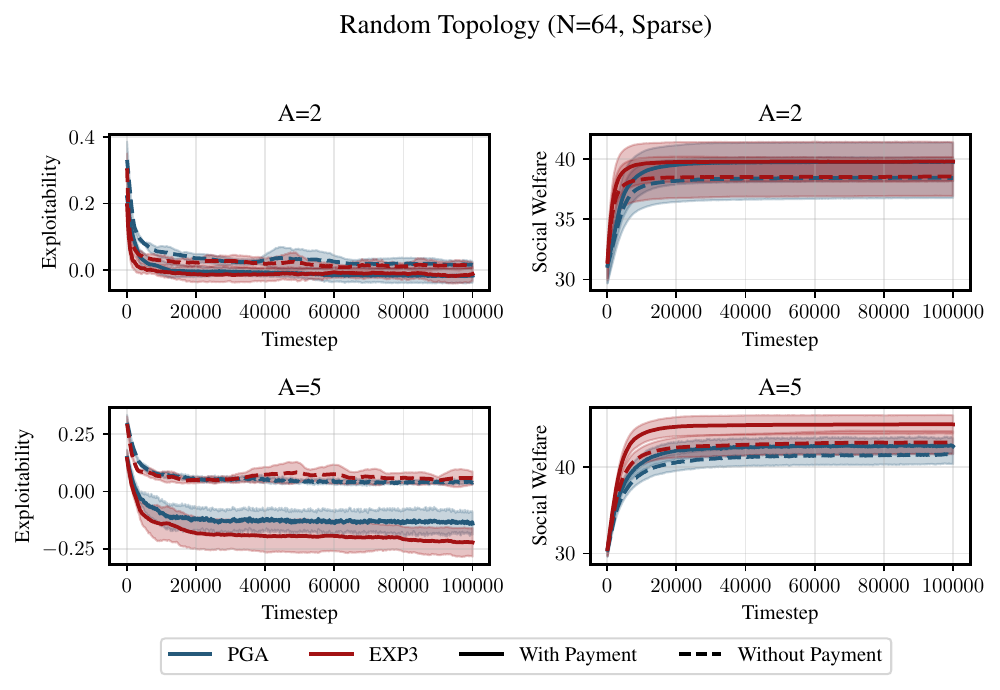}
    \caption{Training curves of different algorithms in a 64-player polymatrix game whose interaction graph follows an Erd\H{o}s--R\'enyi model: each pair of players is connected independently with probability $\frac{3}{64}$. All players share the same action set size $A$.}
    \label{fig:N-64-sparse}
\end{figure}

\section{Conclusion}
\label{section:conclusion}

In this paper, we introduce \emph{Self-Enforcing Transfer Equilibrium} (SETE), an equilibrium notion in which players may commit to peer-to-peer transfers that are paid only if the recipient does not deviate from a prescribed strategy. In polymatrix games, we show in \Cref{lemma:social-welfare} that every stationary point of the social-welfare function induces a \Cref{def:SETE}, given an appropriate payment plan. This gives a mediator-free, budget-balanced way to compute and learn product-form equilibria, and \Cref{theorem:agent-normal-form} shows that such SETE outcomes are Nash equilibria in the agent normal form of the augmented game. We also identify the limitation of this guarantee: without additional enforcement, a player may profit by jointly deviating across the payment and action stages.

\Cref{section:mediated-SETE} resolves this limitation by introducing \emph{Mediated Self-Enforcing Transfer Equilibrium} (M-SETE). When a mediator can make both the payment schedule and the prescribed strategies binding offers, any M-SETE is a Nash equilibrium of the augmented game itself, not merely of its agent normal form. Moreover, M-SETE applies to arbitrary finite games and can support any socially optimal strategy profile while preserving budget balance. Finally, for SETE in polymatrix games, we provide efficient computation and a decentralized learning dynamic (\Cref{alg:learning}) under which players converge to SETE using only realized utilities and received payments.

\section*{Acknowledgement}

We thank Zachary Wojtowicz and Runyu Zhang for valuable discussions. The authors are grateful to the anonymous reviewers and the area chair for their constructive feedback. M.L. was supported by the MathWorks Fellowship. A.O. and G.F. were supported in part by the ONR grant N000142512296. G.F. was additionally supported by CCF-2443068 and an AI2050 Early Career Fellowship.

\bibliographystyle{plainnat}
\bibliography{main}

\newpage
\appendix

\section{Omitted Proofs in \Cref{section:concept}}
\label{section:concept-appendix}

\restate{lemma:variation-to-exploitability}

\begin{proof}
    For any player $i\in [N]$ and any $\hat a_i\in\cA_i$, let $\hat \pi_i\in\Delta^{\cA_i}$ be the pure strategy with $\hat \pi_i(\hat a_i)=1$. Since $SW$ is multilinear in each player's strategy,
    \begin{align*}
        \sum_{j=1}^N \EE_{\ba\sim\pi} \sbr{\cU_j\rbr{\hat a_i, \ba_{-i}}} - SW(\pi)=\inner{\nabla_i SW(\pi)}{\hat \pi_i - \pi_i}.
    \end{align*}
    Because the feasible set is a product of simplexes, the Euclidean projection defining $\pi'$ decomposes across players. In particular, $\pi_i'=\Proj{\Delta^{\cA_i}}{\pi_i+\lambda\nabla_i SW(\pi)}$. In \Cref{lemma:three-point}, take
    \begin{align*}
        \cX=\Delta^{\cA_i},\quad \bx^{(2)} = \hat \pi_i, \quad \bx^{(1)}=\pi_i',\quad \bx^{(0)} = \pi_i,\quad \bu = \nabla_i SW(\pi).
    \end{align*}
    Then
    \begin{align*}
        \lambda\inner{\nabla_i SW(\pi)}{\hat \pi_i - \pi_i'}
        \leq& \frac{1}{2}\rbr{\nbr{\hat \pi_i - \pi_i}^2 - \nbr{\hat \pi_i - \pi_i'}^2 - \nbr{\pi_i' - \pi_i}^2}\\
        =&\inner{\hat \pi_i-\pi_i'}{\pi_i' - \pi_i}\\
        \overset{(i)}{\leq}& \nbr{\hat \pi_i-\pi_i'}\cdot\nbr{\pi_i' - \pi_i}\\
        \overset{(ii)}{\leq}& \sqrt{2}\nbr{\pi_i' - \pi_i}.
    \end{align*}
    $(i)$ is by \holder. $(ii)$ is because $\hat \pi_i, \pi_i'\in\Delta^{\cA_i}$. Therefore,
    \begin{align*}
        \inner{\nabla_i SW(\pi)}{\hat \pi_i - \pi_i}
        \leq& \frac{\sqrt 2}{\lambda}\nbr{\pi_i' - \pi_i}
        + \inner{\nabla_i SW(\pi)}{\pi_i' - \pi_i}\\
        \leq& \frac{\sqrt 2}{\lambda}\nbr{\pi_i' - \pi_i}
        + \nbr{\nabla_i SW(\pi)}\cdot \nbr{\pi_i' - \pi_i}.
    \end{align*}
    Since each coordinate of $\nabla_i SW(\pi)$ lies in $[0,N]$, we have $\nbr{\nabla_i SW(\pi)}\leq N\sqrt A$. The assumption $\lambda\leq 1/(NA)$ implies $\lambda\nbr{\nabla_i SW(\pi)}\leq 1$, and hence
    \begin{align*}
        \inner{\nabla_i SW(\pi)}{\hat \pi_i - \pi_i}
        \leq \frac{\sqrt 2+1}{\lambda}\nbr{\pi_i' - \pi_i}
        \leq \frac{3}{\lambda}\sqrt{\sum_{j=1}^N \nbr{\pi_j' - \pi_j}^2}.
    \end{align*}
    Maximizing over $i$ and $\hat a_i$ gives the claim. \qedhere
\end{proof}

\restate{lemma:sw-optimal-not-NE}

\begin{proof}
As shown in \Cref{table:PD-variant}, consider the unique socially optimal strategy profile $\pi$ that plays $\rbr{C,C}$ with probability one. Let $p_{r\to c}$ denote the transfer from the row player to the column player, and let $p_{c\to r}$ denote the transfer in the opposite direction. 

\begin{table}[h]
\centering
\begin{tabular}{|c|c|c|}
    \hline
     & \textbf{Cooperate (C)} & \textbf{Defect (D)} \\
    \hline
    \textbf{Cooperate (C)} & $(0.6,0.6)$ & $(0.4,0.55)$ \\
    \hline
    \textbf{Defect (D)} & $(1,0)$ & $(0.5,0.5)$ \\
    \hline
    \end{tabular}
    \captionof{table}{Utility matrix for a variation of the Prisoner's Dilemma. Each entry $(a,b)$ gives the payoffs to the row player ($a$) and the column player ($b$).}
    \label{table:PD-variant}
\end{table}

For $\rbr{\pi,\bp}$ to be a Nash equilibrium of the augmented game, the following constraints must hold:
\begin{align}
    p_{r\to c} \leq&\, 0.6-0.4 = 0.2, \label{eq:minimax-1}\\
    p_{c\to r} \geq&\, 1-0.6 = 0.4, \label{eq:minimax-2}\\
    p_{c\to r}-p_{r\to c} \leq&\, 0.6-0.5 = 0.1. \label{eq:minimax-3}
\end{align}
To justify \Cref{eq:minimax-1}, consider the (off-path) deviation in which the row player withdraws their outgoing transfer $p_{r\to c}$. After withdrawing, the row player's payoff in the underlying game is bounded below by $0.4$, regardless of the continuation play: indeed,
\begin{align*}
\min_{a_c\in\cbr{C,D}} u_r\!\rbr{C,a_c} &= 0.4
\quad\text{and}\quad
\min_{a_c\in\cbr{C,D}} u_r\!\rbr{D,a_c} = 0.5,
\end{align*}
so in particular, the row player can guarantee at least $0.4$ no matter what the column player does. Since under $\pi$ the row player obtains $0.6$ before transfers, their willingness to pay cannot exceed $0.6-0.4=0.2$, which yields \Cref{eq:minimax-1}.

Next, \Cref{eq:minimax-2} follows from \Cref{eq:stable}: if the row player deviates from $C$ to $D$ while the column player plays $C$, the row player's payoff increases from $0.6$ to $1$, so deterring this deviation requires incoming transfers of at least $1-0.6$.

Finally, \Cref{eq:minimax-3} is necessary because the column player can always deviate by withdrawing all outgoing payments and playing $D$. Regardless of the row player's behavior, this guarantees the column player an expected payoff of at least $0.5$ (since the payoff from $D$ is $0.55$ against $C$ and $0.5$ against $D$). Therefore, the equilibrium payoff under $\rbr{C,C}$,
$0.6 - p_{c\to r} + p_{r\to c}$, must be at least $0.5$, which is equivalent to \Cref{eq:minimax-3}.

However, no payment plan $\bp$ can satisfy these constraints simultaneously, because
\begin{equation*}
0.2 \overset{\Cref{eq:minimax-1}}{\ge}\, p_{r\to c}
\overset{\Cref{eq:minimax-3}}{\ge}\, p_{c\to r}-0.1
\overset{\Cref{eq:minimax-2}}{\ge}\, 0.3. \qedhere
\end{equation*}
\end{proof}

\section{Omitted Proofs in \Cref{section:polymatrix}}
\label{section:auxiliary-convergence-sw}

\begin{lemma}
\label{lemma:polymatrix-optimal-sw-hard}
    Computing the optimal social welfare in polymatrix games is {\sf NP}-hard.
\end{lemma}

\begin{proof}
    The proof is motivated by \citet[Theorem 4.3]{liu2026-MASE}. We reduce from the $k$-coloring problem, where $k \geq 3$ is fixed. Given any graph $G = (V,E)$, construct a polymatrix game with one player for each vertex in $V$. Let $n = \abr{V}$, and let every player have action set $\cbr{1,2,\dots,k}$, where each action represents a color. The neighborhood of player $i$ is the same as the neighborhood of the corresponding vertex in $G$, namely
    \begin{align*}
        \cN(i) = \cbr{j \in [n] \colon (i,j) \in E}.
    \end{align*}
    For every edge $(i,j) \in E$, define the pairwise utility by
    \begin{align*}
        \cU_{i,j}(a_i,a_j) = \frac{\ind\rbr{a_i \neq a_j}}{n}.
    \end{align*}
    For a pure action profile, the utility of player $i$ is at most $\abr{\cN(i)}/n$. Therefore, the social welfare is at most
    \begin{align*}
        \frac{1}{n} \sum_{i=1}^n \abr{\cN(i)}
        =
        \frac{2\abr{E}}{n}.
    \end{align*}
    It is enough to consider pure action profiles, since the social welfare of any mixed profile is a convex combination of the social welfare values of pure profiles. The upper bound above is attained exactly when every edge has endpoints with different colors. Hence, the optimal social welfare equals $2\abr{E}/n$ if and only if $G$ admits a valid $k$ coloring. Thus, an algorithm for computing the optimal social welfare in polymatrix games would decide the $k$-coloring problem. Since the $k$-coloring problem is {\sf NP} complete for every fixed $k \geq 3$, computing the optimal social welfare in polymatrix games is {\sf NP}-hard. \qedhere
\end{proof}

\restate{lemma:smoothness}

\begin{proof}
    For any differentiable function $f\colon\cX\to\RR$, to prove \Cref{eq:smoothness}, we only need to prove that its gradient is $2\beta$-Lipschitz continuous, \emph{i.e.}, for any $\bx,\bx'\in\cX$, we have
    \begin{align*}
        \nbr{\nabla f(\bx') - \nabla f(\bx)}\leq 2\beta\nbr{\bx'-\bx}.
    \end{align*}
    Because
    \begin{align*}
        \abr{f(\bx') - f(\bx) - \inner{\nabla f(\bx)}{\bx' - \bx}} =& \abr{\int_0^1 \inner{\nabla f(\bx + t(\bx'-\bx))}{\bx'-\bx}{\rm d}t -\inner{\nabla f(\bx)}{\bx' - \bx} }\\
        =&\abr{\int_0^1 \inner{\nabla f(\bx + t(\bx'-\bx)) - \nabla f(\bx)}{\bx'-\bx}{\rm d}t }\\
        \overset{(i)}{\leq}& \abr{\int_0^1 \nbr{\nabla f(\bx + t(\bx'-\bx)) - \nabla f(\bx)}\cdot\nbr{\bx'-\bx}{\rm d}t }\\
        \leq& \abr{\int_0^1 2t\beta \nbr{\bx'-\bx}^2{\rm d}t }\\
        =& \beta \nbr{\bx'-\bx}^2.
    \end{align*}
    $(i)$ is by \holder. Therefore, in the following, we will show that
    \begin{align*}
        \nbr{\nabla SW(\pi') - \nabla SW(\pi)}\leq 2NA \nbr{\pi'-\pi}=2\beta\nbr{\pi'-\pi}.
    \end{align*}
    By the polymatrix decomposition, for any player $i\in [N]$ and action $a_i\in\cA_i$,
    \begin{align*}
        \frac{\partial SW(\pi)}{\partial \pi_i(a_i)}
        =&\sum_{j\in\cN(i)}\sum_{a_j\in\cA_j}\cU_{i,j}(a_i,a_j)\pi_j(a_j)
        +\sum_{k:\,i\in\cN(k)}\sum_{a_k\in\cA_k}\cU_{k,i}(a_k,a_i)\pi_k(a_k).
    \end{align*}
    Hence, since each pairwise payoff is bounded in $[0,1]$,
    \begin{align*}
        \abr{\frac{\partial SW(\pi')}{\partial \pi_i'(a_i)} - \frac{\partial SW(\pi)}{\partial \pi_i(a_i)}}
        \leq& \sum_{j\in\cN(i)}\nbr{\pi_j'-\pi_j}_1 + \sum_{k:\,i\in\cN(k)}\nbr{\pi_k'-\pi_k}_1\\
        \leq& 2\sum_{j\in[N]\setminus\cbr{i}}\nbr{\pi_j'-\pi_j}_1.
    \end{align*}
    Therefore, by letting $A\coloneqq\max_{i\in [N]}|\cA_i|$ denote the size of the largest action set, we have
    \begin{align*}
        \nbr{\nabla SW(\pi') - \nabla SW(\pi)}^2
        \leq& 4\sum_{i=1}^N |\cA_i|\rbr{\sum_{j\in[N]\setminus\cbr{i}}\nbr{\pi_j'-\pi_j}_1}^2\\
        \overset{(i)}{\leq}& 4N A\sum_{i=1}^N\sum_{j\in[N]\setminus\cbr{i}}\nbr{\pi_j'-\pi_j}_1^2\\
        \overset{(ii)}{\leq}& 4N A\sum_{i=1}^N\sum_{j\in[N]\setminus\cbr{i}}|\cA_j|\nbr{\pi_j'-\pi_j}^2\\
        \leq& 4N^2A^2\nbr{\pi'-\pi}^2.
    \end{align*}
    $(i)$ and $(ii)$ are both by Cauchy-Schwarz inequality.\qedhere
\end{proof}

\restate{lemma:three-point}

\begin{proof}
    By the definition of projection, we have
    \begin{align*}
        \bx^{(1)}=\argmin_{\bx\in\cX} \frac{1}{2} \nbr{\bx - \rbr{\bx^{(0)} + \lambda \bu}}^2.
    \end{align*}
    Then, by first-order optimality,
    \begin{align*}
        \inner{\rbr{\bx^{(1)} - \bx^{(0)}} - \lambda \bu}{\bx^{(2)}-\bx^{(1)}}\geq 0.
    \end{align*}
    Hence,
    \begin{align*}
        \lambda\inner{\bu}{\bx^{(2)}-\bx^{(1)}}\leq& \inner{\bx^{(1)} - \bx^{(0)}}{\bx^{(2)}-\bx^{(1)}}\\
        =& \frac{1}{2}\rbr{\nbr{\bx^{(2)} - \bx^{(0)}}^2 - \nbr{\bx^{(2)} - \bx^{(1)}}^2 - \nbr{\bx^{(1)} - \bx^{(0)}}^2}. \qedhere
    \end{align*}
\end{proof}

\restate{lemma:sw-to-SETE-gap}

\begin{proof}
    By definition of the payment $\overbar \bp^{(t)}$ in \Cref{eq:update-rule-potential-game}, for any player $i\in [N]$ and any action $\hat a_i\in\cA_i$, we have
    \begin{align*}
        &\EE_{\ba\sim\overbar \pi^{(t)}} \sbr{\cU_i\rbr{\hat a_i, \ba_{-i}}} -\rbr{ \EE_{\ba\sim\overbar \pi^{(t)}} \sbr{\cU_i\rbr{\ba}} + \sum_{\substack{j=1,\\j\neq i}}^N \overbar p_{j\to i}^{(t)} }\\
        =&\EE_{\ba\sim\overbar \pi^{(t)}} \sbr{\cU_i\rbr{\hat a_i, \ba_{-i}}} -\rbr{ \EE_{\ba\sim\overbar \pi^{(t)}} \sbr{\cU_i\rbr{\ba}} + \sum_{\substack{j=1,\\j\neq i}}^N \EE_{\ba\sim\overbar \pi^{(t)}}\sbr{\cU_j\rbr{\ba}} - \min_{\hat a_i^{\,\prime}\in\cA_i} \EE_{\ba\sim\overbar \pi^{(t)}}\sbr{\cU_j\rbr{\hat a_i^{\,\prime}, \ba_{-i}}} }\\
        =&\EE_{\ba\sim\overbar \pi^{(t)}} \sbr{\cU_i\rbr{\hat a_i, \ba_{-i}}} + \sum_{\substack{j=1,\\j\neq i}}^N \min_{\hat a_i^{\,\prime}\in\cA_i} \EE_{\ba\sim\overbar \pi^{(t)}}\sbr{\cU_j\rbr{\hat a_i^{\,\prime}, \ba_{-i}}} - SW(\overbar \pi^{(t)}).
    \end{align*}
    Furthermore,
    \begin{align*}
        &\EE_{\ba\sim\overbar \pi^{(t)}} \sbr{\cU_i\rbr{\hat a_i, \ba_{-i}}} + \sum_{\substack{j=1,\\j\neq i}}^N \min_{\hat a_i^{\,\prime}\in\cA_i} \EE_{\ba\sim\overbar \pi^{(t)}}\sbr{\cU_j\rbr{\hat a_i^{\,\prime}, \ba_{-i}}} - SW(\overbar \pi^{(t)})\\
        \leq&\EE_{\ba\sim\overbar \pi^{(t)}} \sbr{\cU_i\rbr{\hat a_i, \ba_{-i}}} + \sum_{\substack{j=1,\\j\neq i}}^N \EE_{\ba\sim\overbar \pi^{(t)}}\sbr{\cU_j\rbr{\hat a_i, \ba_{-i}}} - SW(\overbar \pi^{(t)})\\
        =&\sum_{j=1}^N \EE_{\ba\sim\overbar \pi^{(t)}} \sbr{\cU_j\rbr{\hat a_i, \ba_{-i}}} - SW(\overbar \pi^{(t)})\\
        \leq& \max_{\hat a_i^{\,\prime}\in\cA_i}~~ \sum_{j=1}^N \EE_{\ba\sim\overbar \pi^{(t)}} \sbr{\cU_j\rbr{\hat a_i^{\,\prime}, \ba_{-i}}} - SW(\overbar \pi^{(t)})\\
        \leq& \max_{\substack{k\in [N]\\\hat a_k\in\cA_k}}~~ \sum_{j=1}^N \EE_{\ba\sim\overbar \pi^{(t)}} \sbr{\cU_j\rbr{\hat a_k, \ba_{-k}}} - SW(\overbar \pi^{(t)}).
    \end{align*}
    Because $i$ and $\hat a_i$ can be chosen arbitrarily, the proof is concluded. \qedhere
\end{proof}

\begin{lemma}
\label{lemma:sw-grad-u}
    For any strategy profile $\pi\in \bigtimes_{i=1}^N \Delta^{\cA_i}$, any player $i\in [N]$, and action $a_i\in\cA_i$, we have
    \begin{align*}
        \frac{\partial SW(\pi)}{\partial \pi_i(a_i)} = \sum_{j=1}^N \EE_{\ba'\sim\pi}\sbr{\cU_j(a_i, \ba_{-i}')}.
    \end{align*}
\end{lemma}

\begin{proof}
    For any player $i\in [N]$ and action $a_i\in\cA_i$, we have
\begin{align*}
    \frac{\partial SW(\pi)}{\partial \pi_i(a_i)} =& \frac{\partial \rbr{\sum_{\ba'\in\cA}\sum_{j=1}^N \cU_j(\ba')\prod_{k=1}^N \pi_k(a_k')}}{\partial \pi_i(a_i)}\\
    =&\sum_{\ba_{-i}'\in\cA_{-i}}\sum_{j=1}^N \cU_j(a_i, \ba_{-i}')\prod_{k\in[N]\setminus\cbr{i}} \pi_k(a_k')\\
    \overset{(i)}{=}&\rbr{\sum_{a_i'\in\cA_i} \pi_i(a_i')}\sum_{\ba_{-i}'\in\cA_{-i}}\sum_{j=1}^N \cU_j(a_i, \ba_{-i}')\prod_{k\in[N]\setminus\cbr{i}} \pi_k(a_k')\\
    =&\sum_{\ba'\in\cA}\sum_{j=1}^N \cU_j(a_i, \ba_{-i}')\prod_{k=1}^N \pi_k(a_k')\\
    =&\sum_{j=1}^N \EE_{\ba'\sim\pi}\sbr{\cU_j(a_i, \ba_{-i}')}.
\end{align*}
$(i)$ uses the fact that $\sum_{a_i'\in\cA_i} \pi_i(a_i')=1$ since $\pi_i\in\Delta^{\cA_i}$.
\end{proof}

\section{Proof of Auxiliary Lemmas for \Cref{theorem:convergence-sw}}
\label{section:bandit-auxiliary}

\restate{lemma:unbiased-estimator}

\begin{proof}
    For any player $i\in [N]$, any timestep $t\in [T]$, and any action $a_i\in \cA_i$, we have
\begin{align*}
    &\EE\sbr{\tilde u_i^{(t)}(a_i)\given \pi^{(t)}}\\
    =&\sum_{\ba_{-i}\in\cA_{-i}} \rbr{\frac{1}{\pi_i^{e, (t)}(a_i)}\sum_{j\in\cN(i)} \rbr{\cU_{i, j}\rbr{a_i, a_j} + \rbr{\cU_{j, i}\rbr{a_j, a_i}  - \min_{\hat a_i\in\cA_i} \cU_{j, i}\rbr{a_j, \hat a_i}}}} \prod_{k=1}^N \pi_k^{e, (t)}(a_k)\\
    =&\sum_{\ba_{-i}\in\cA_{-i}} \sum_{j\in\cN(i)}\rbr{\cU_{i, j}\rbr{a_i, a_j} + \cU_{j, i}\rbr{a_j, a_i} - \min_{\hat a_i\in\cA_i} \cU_{j, i}\rbr{a_j, \hat a_i}} \prod_{k\in[N]\setminus\cbr{i}} \pi_k^{e, (t)}(a_k)\\
    =&\sum_{j\in\cN(i)} \sum_{a_j\in\cA_j} \rbr{\cU_{i, j}\rbr{a_i, a_j} + \cU_{j, i}\rbr{a_j, a_i} - \min_{\hat a_i\in\cA_i} \cU_{j, i}\rbr{a_j, \hat a_i}} \pi_j^{e, (t)}(a_j).
\end{align*}

For any player $i\in [N]$ and any action $a_i\in\cA_i$, we have
\begin{align*}
    \frac{\partial SW(\pi^{e, (t)})}{\partial \pi_i^{e, (t)}(a_i)} =& \frac{\partial \rbr{\sum_{j=1}^N \sum_{k\in \cN(j)} \sum_{\substack{a_j\in\cA_j,\\a_k\in\cA_k}} \cU_{j, k}(a_j, a_k)\pi_j^{e, (t)}(a_j)\pi_k^{e, (t)}(a_k)}}{\partial \pi_i^{e, (t)}(a_i)}\\
    \overset{(i)}{=}&\frac{\partial \rbr{\sum_{j\in\cN(i)} \sum_{a_j\in\cA_j} \rbr{\cU_{i, j}(a_i, a_j) + \cU_{j, i}(a_j, a_i)}\pi_j^{e, (t)}(a_j)\pi_i^{e, (t)}(a_i)}}{\partial \pi_i^{e, (t)}(a_i)}\\
    =&\sum_{j\in\cN(i)} \sum_{a_j\in\cA_j} \rbr{\cU_{i, j}(a_i, a_j) + \cU_{j, i}(a_j, a_i)}\pi_j^{e, (t)}(a_j).
\end{align*}
$(i)$ follows because only the utilities of player $i$ and $i$'s neighbors $j \in \cN(i)$ depend on $\pi_i^{e, (t)}(a_i)$.
Therefore,
\begin{equation*}
    \EE\sbr{\tilde u_i^{(t)}(a_i) \given \pi^{(t)}} = \frac{\partial SW(\pi^{e, (t)})}{\partial \pi_i^{e, (t)}(a_i)} - \sum_{j\in\cN(i)} \sum_{a_j\in\cA_j} \rbr{\min_{\hat a_i\in\cA_i} \cU_{j, i}\rbr{a_j, \hat a_i}} \pi_j^{e, (t)}(a_j). \qedhere
\end{equation*}
\end{proof}

\restate{lemma:bounded-variance}

\begin{proof}
    By definition, for any player $i\in [N]$,
\begin{align*}
    \EE\sbr{\nbr{\tilde\bu_i^{(t)}}^2\given \pi^{(t)}}
    =& \sum_{a_i\in\cA_i}\pi_i^{e,(t)}(a_i)\,
    \EE\sbr{\rbr{\frac{\sum_{j\in\cN(i)} \rbr{\cU_{i, j}\rbr{a_i, a_j^{(t)}} + \tilde p_{j\to i}^{(t)}}}{\pi_i^{e,(t)}(a_i)}}^2\given a_i^{(t)}=a_i,\pi^{(t)}}.
\end{align*}
Since $\pi_i^{e, (t)}(a_i)\geq \frac{\gamma}{|\cA_i|}$ for any $a_i\in\cA_i$, $\sum_{j\in\cN(i)} \cU_{i, j} \rbr{a_i, a_j^{(t)}} = \cU_i(a_i,\ba_{-i}^{(t)})\in [0, 1]$, and
\begin{align*}
    \tilde p_{j\to i}^{(t)} =~ \cU_{j, i}\rbr{a_j^{(t)}, a_i}  - \min_{\hat a_i\in\cA_i} \cU_{j, i}\rbr{a_j^{(t)}, \hat a_i} \in [0, 1],
\end{align*}
we have
\begin{align*}
    \EE\sbr{\nbr{\tilde\bu^{(t)}}^2} = \sum_{i=1}^N \EE\sbr{\nbr{\tilde\bu_i^{(t)}}^2}\leq& \sum_{i=1}^N \sum_{a_i\in\cA_i} \frac{\rbr{|\cN(i)|+1}^2}{\pi_i^{e,(t)}(a_i)}\leq \sum_{i=1}^N \frac{|\cA_i|^2}{\gamma} \rbr{|\cN(i)| + 1}^2\overset{(i)}{\leq} \frac{A^2}{\gamma} N^3,
\end{align*}
where $A=\max_{i\in [N]}|\cA_i|$. $(i)$ is because $|\cN(i)|+1\leq N$ for all $i\in [N]$. \qedhere
\end{proof}

\restate{lemma:bandit-distance-NE}

\begin{proof}
    By \Cref{lemma:smoothness}, $SW(\mu)-2\beta\nbr{\mu-\pi}^2$ is strongly concave with respect to $\mu$, so the maximizer is unique. By first-order optimality, $\mu$ maximizes $SW(\mu)-2\beta\nbr{\mu-\pi}^2$ if and only if: for any $\mu'\in\bigtimes_{i=1}^N \Delta^{\cA_i}$,
    \begin{align}
        \inner{\nabla SW(\mu) - 4\beta (\mu-\pi)}{\mu' - \mu}\leq 0.\label{eq:first-order-condition-1}
    \end{align}
    Moreover, let
    \begin{align*}
        \mu'' = \argmin_{\mu'\in \bigtimes_{i=1}^N \Delta^{\cA_i}} \nbr{\mu' - \rbr{\pi + \frac{1}{4\beta}\nabla SW(\mu)}}^2.
    \end{align*}
    $\mu''$ is a minimizer if and only if: for any $\mu'\in\bigtimes_{i=1}^N \Delta^{\cA_i}$,
    \begin{align}
         2\inner{\mu'' - \rbr{\pi + \frac{1}{4\beta}\nabla SW(\mu)}}{\mu' - \mu''}\geq 0.\label{eq:first-order-condition-2}
    \end{align}
    Hence, since \Cref{eq:first-order-condition-1} and \Cref{eq:first-order-condition-2} are equivalent, we have
    \begin{align*}
        \mu = \mu'' = \Proj{\bigtimes_{i=1}^N \Delta^{\cA_i}}{\pi + \frac{1}{4\beta}\nabla SW(\mu)}.
    \end{align*}
    Since the objective is linear in a unilateral mixed deviation $\hat\pi_i$, maximizing over $\hat\pi_i\in\Delta^{\cA_i}$ is equivalent to maximizing over pure actions. Thus, by \Cref{lemma:variation-to-exploitability}, we have
    \begin{align*}
        &\max_{\substack{i\in [N]\\\hat\pi_i\in\Delta^{\cA_i}}}\inner{\nabla_i SW(\pi)}{\hat \pi_i - \pi_i}\\
        \leq& 12\beta\nbr{\Proj{\bigtimes_{i=1}^N \Delta^{\cA_i}}{\pi + \frac{1}{4\beta} \nabla SW(\pi)} - \pi}\\
        \leq& 12\beta\nbr{\Proj{\bigtimes_{i=1}^N \Delta^{\cA_i}}{\pi + \frac{1}{4\beta} \nabla SW(\pi)} - \mu} + 12\beta \nbr{\mu - \pi}\\
        =& 12\beta\nbr{\Proj{\bigtimes_{i=1}^N \Delta^{\cA_i}}{\pi + \frac{1}{4\beta} \nabla SW(\pi)} - \Proj{\bigtimes_{i=1}^N \Delta^{\cA_i}}{\pi + \frac{1}{4\beta}\nabla SW(\mu)}} + 12\beta \nbr{\mu - \pi}\\
        \overset{(i)}{\leq}& 3\nbr{\nabla SW(\pi) - \nabla SW(\mu)} + 12\beta \nbr{\mu - \pi}\\
        \overset{(ii)}{\leq}& 18\beta\nbr{\mu - \pi}.
    \end{align*}
    $(i)$ uses the property of projection. $(ii)$ uses \Cref{lemma:smoothness}. \qedhere
\end{proof}

\section{Experiment Details}
\label{section:experiment-details}

This section provides the additional experimental details omitted from \Cref{section:experiment}.

\subsection{Random Polymatrix Game}

We consider random polymatrix games under four underlying graph topologies:
\begin{enumerate}
    \item \textbf{Erd\H{o}s--R\'enyi \& Dense}: The interaction graph is an Erd\H{o}s--R\'enyi graph. For each unordered pair of distinct players $i\neq j$, an edge between $i$ and $j$ is present independently with probability $0.3$.
    \item \textbf{Erd\H{o}s--R\'enyi \& Sparse}: The interaction graph is an Erd\H{o}s--R\'enyi graph. For each unordered pair of distinct players $i\neq j$, an edge is present independently with probability $\frac{3}{N}$, where $N$ is the number of players.
    \item \textbf{Local Connectivity \& Dense}: Players are partitioned into $3$ clusters. Players within the same cluster form a clique. For players in different clusters, cross-cluster edges are added independently with probability $\frac{3}{N}$. This models settings with strong within-group interactions and relatively infrequent interactions across groups.
    \item \textbf{Local Connectivity \& Sparse}: Players are sampled uniformly at random on a two-dimensional plane, and undirected edges are added between each player and its $3$ nearest neighbors. This models a transportation-like network in which physical distance restricts interactions.
\end{enumerate}

In the experiments, all players have the same action set size, that is, $\abr{\cA_i}=A$ for every $i\in[N]$. For each player $i\in[N]$, we first sample a player dependent mean matrix $M_i\in[0,1]^{A\times A}$ uniformly. Then, for each neighbor $j\in\cN(i)$ and each action pair $(a_i,a_j)\in\cA_i\times\cA_j$, we sample the corresponding entry of the pairwise utility matrix independently as
\begin{align*}
    \cU_{i,j}(a_i,a_j)\sim
    \begin{cases}
        {\rm Uniform}\rbr{0, 2M_i(a_i,a_j)} & M_i(a_i,a_j) \leq 0.5,\\
        {\rm Uniform}\rbr{2M_i(a_i,a_j)-1, 1} & M_i(a_i,a_j) > 0.5.
    \end{cases}
\end{align*}
This construction guarantees that
\begin{align*}
    \EE\sbr{\cU_{i,j}(a_i,a_j)\mid M_i}=M_i(a_i,a_j),
\end{align*}
while keeping every entry in $[0,1]$. After all pairwise utility matrices are sampled, we rescale them so that the total utility of each player is bounded by $1$, namely
\begin{align*}
    \max_{\ba\in\cA}\sum_{j\in\cN(i)}\cU_{i,j}(a_i,a_j)\leq 1
\end{align*}
for every $i\in[N]$.

The purpose of sampling the mean matrix $M_i$ is to make the aggregate utility values associated with different action pairs distinguishable. If all entries of each $\cU_{i,j}$ were instead sampled independently with a common mean, then, when $\abr{\cN(i)}$ is large, the aggregate utilities $\sum_{j\in\cN(i)}\cU_{i,j}$ would concentrate around similar values. In that regime, different strategies yield nearly identical utilities, making the empirical performance differences among algorithms less informative.

\subsection{Algorithm Details}

The only difference between running an algorithm \emph{with} versus \emph{without} payments is the construction of the estimated utility signal. Specifically, upon playing action $a_i^{(t)}$, player $i$ forms
\begin{align}
    \tilde u_i^{(t)}(a_i^{(t)})\gets\begin{cases}
                    \frac{1}{\pi_i^{e,(t)}(a_i^{(t)})}\sum_{j\in\cN(i)} \rbr{\cU_{i, j}\rbr{a_i^{(t)}, a_j^{(t)}} + \tilde p_{j\to i}^{(t)}} & \text{With Payment},\\[4pt]\\
                    \frac{1}{\pi_i^{e,(t)}(a_i^{(t)})}\sum_{j\in\cN(i)} \cU_{i, j}\rbr{a_i^{(t)}, a_j^{(t)}} & \text{Without Payment}.
                \end{cases}
\end{align}

For PGA, the variant with payments coincides with \Cref{alg:learning}. The variant without payments is obtained by removing the payment terms from \Cref{alg:learning}, which recovers the classical PGA dynamics \citep{hazan2016introduction}. We set the learning rate and exploration parameter to $\eta=T^{-2/3}$ and $\gamma=T^{-1/3}$, respectively.

For EXP3, we use the same learning rate in both variants, namely $\eta=T^{-1/2}$.

\subsection{More Experimental Results}

In \Cref{table:last_iterate_exploitability,table:last_iterate_social_welfare}, we report, respectively, the last-iterate exploitability and social welfare of each algorithm on polymatrix games with varying numbers of players, graph topologies, and action set sizes.

\begin{longtable}{lllll}
\caption{Last-iterate exploitability.}\label{table:last_iterate_exploitability}\\
\toprule
Topology & PGA (pay) & PGA (no pay) & EXP3 (pay) & EXP3 (no pay)\\
\midrule
\endfirsthead
\toprule
Topology & PGA (pay) & PGA (no pay) & EXP3 (pay) & EXP3 (no pay)\\
\midrule
\endhead
\bottomrule
\endfoot
\multicolumn{5}{l}{\textbf{N=16, dense, $A$=2}}\\
random & -0.0173 $\pm$ 0.0234 & 0.0071 $\pm$ 0.0051 & -0.0136 $\pm$ 0.0235 & 0.0080 $\pm$ 0.0141\\
local & -0.0441 $\pm$ 0.0369 & 0.0043 $\pm$ 0.0017 & -0.0412 $\pm$ 0.0342 & 0.0019 $\pm$ 0.0014\\
\addlinespace
\multicolumn{5}{l}{\textbf{N=16, dense, $A$=5}}\\
random & -0.1023 $\pm$ 0.0501 & 0.0373 $\pm$ 0.0203 & -0.1048 $\pm$ 0.0605 & 0.0356 $\pm$ 0.0317\\
local & -0.1389 $\pm$ 0.0409 & 0.0217 $\pm$ 0.0047 & -0.1546 $\pm$ 0.0524 & 0.0124 $\pm$ 0.0105\\
\addlinespace
\multicolumn{5}{l}{\textbf{N=64, dense, $A$=2}}\\
random & -0.0556 $\pm$ 0.0160 & 0.0078 $\pm$ 0.0032 & -0.0532 $\pm$ 0.0149 & 0.0061 $\pm$ 0.0048\\
local & -0.0591 $\pm$ 0.0155 & 0.0069 $\pm$ 0.0016 & -0.0562 $\pm$ 0.0171 & 0.0045 $\pm$ 0.0031\\
\addlinespace
\multicolumn{5}{l}{\textbf{N=64, dense, $A$=5}}\\
random & -0.1083 $\pm$ 0.0287 & 0.0220 $\pm$ 0.0036 & -0.1201 $\pm$ 0.0370 & 0.0173 $\pm$ 0.0124\\
local & -0.1500 $\pm$ 0.0239 & 0.0228 $\pm$ 0.0033 & -0.1726 $\pm$ 0.0388 & 0.0115 $\pm$ 0.0039\\
\addlinespace
\multicolumn{5}{l}{\textbf{N=256, dense, $A$=2}}\\
random & -0.0814 $\pm$ 0.0088 & 0.0065 $\pm$ 0.0008 & -0.0731 $\pm$ 0.0129 & 0.0040 $\pm$ 0.0009\\
local & -0.0967 $\pm$ 0.0175 & 0.0090 $\pm$ 0.0033 & -0.0908 $\pm$ 0.0239 & 0.0045 $\pm$ 0.0013\\
\addlinespace
\multicolumn{5}{l}{\textbf{N=256, dense, $A$=5}}\\
random & -0.1683 $\pm$ 0.0078 & 0.0237 $\pm$ 0.0024 & -0.1873 $\pm$ 0.0103 & 0.0089 $\pm$ 0.0014\\
local & -0.1879 $\pm$ 0.0158 & 0.0280 $\pm$ 0.0030 & -0.2197 $\pm$ 0.0280 & 0.0084 $\pm$ 0.0021\\
\addlinespace
\multicolumn{5}{l}{\textbf{N=16, sparse, $A$=2}}\\
random & -0.0015 $\pm$ 0.0092 & 0.0099 $\pm$ 0.0117 & 0.0005 $\pm$ 0.0089 & 0.0063 $\pm$ 0.0074\\
local & -0.0186 $\pm$ 0.0262 & 0.0092 $\pm$ 0.0095 & -0.0224 $\pm$ 0.0241 & 0.0061 $\pm$ 0.0062\\
\addlinespace
\multicolumn{5}{l}{\textbf{N=16, sparse, $A$=5}}\\
random & -0.0340 $\pm$ 0.0467 & 0.0355 $\pm$ 0.0311 & -0.0281 $\pm$ 0.0478 & 0.0179 $\pm$ 0.0140\\
local & -0.1187 $\pm$ 0.0400 & 0.0283 $\pm$ 0.0104 & -0.1240 $\pm$ 0.0432 & 0.0319 $\pm$ 0.0229\\
\addlinespace
\multicolumn{5}{l}{\textbf{N=64, sparse, $A$=2}}\\
random & -0.0165 $\pm$ 0.0211 & 0.0163 $\pm$ 0.0097 & -0.0112 $\pm$ 0.0218 & 0.0131 $\pm$ 0.0094\\
local & 0.0028 $\pm$ 0.0030 & 0.0142 $\pm$ 0.0085 & 0.0041 $\pm$ 0.0033 & 0.0435 $\pm$ 0.0254\\
\addlinespace
\multicolumn{5}{l}{\textbf{N=64, sparse, $A$=5}}\\
random & -0.1401 $\pm$ 0.0486 & 0.0396 $\pm$ 0.0061 & -0.2221 $\pm$ 0.0625 & 0.0585 $\pm$ 0.0291\\
local & -0.0976 $\pm$ 0.0448 & 0.0489 $\pm$ 0.0162 & -0.0939 $\pm$ 0.0495 & 0.0678 $\pm$ 0.0264\\
\addlinespace
\multicolumn{5}{l}{\textbf{N=256, sparse, $A$=2}}\\
random & -0.0019 $\pm$ 0.0072 & 0.0323 $\pm$ 0.0236 & -0.0012 $\pm$ 0.0104 & 0.0396 $\pm$ 0.0405\\
local & 0.0056 $\pm$ 0.0066 & 0.0308 $\pm$ 0.0148 & 0.0093 $\pm$ 0.0045 & 0.0354 $\pm$ 0.0184\\
\addlinespace
\multicolumn{5}{l}{\textbf{N=256, sparse, $A$=5}}\\
random & -0.1274 $\pm$ 0.0365 & 0.0576 $\pm$ 0.0082 & -0.1790 $\pm$ 0.0590 & 0.0872 $\pm$ 0.0229\\
local & -0.0602 $\pm$ 0.0252 & 0.0534 $\pm$ 0.0155 & -0.0601 $\pm$ 0.0277 & 0.0933 $\pm$ 0.0351\\
\addlinespace
\end{longtable}

\begin{longtable}{lllll}
\caption{Last-iterate social welfare.}\label{table:last_iterate_social_welfare}\\
\toprule
Topology & PGA (pay) & PGA (no pay) & EXP3 (pay) & EXP3 (no pay)\\
\midrule
\endfirsthead
\toprule
Topology & PGA (pay) & PGA (no pay) & EXP3 (pay) & EXP3 (no pay)\\
\midrule
\endhead
\bottomrule
\endfoot
\multicolumn{5}{l}{\textbf{N=16, dense, $A$=2}}\\
random & 7.4327 $\pm$ 0.9532 & 7.2237 $\pm$ 0.9959 & 7.4292 $\pm$ 0.9594 & 7.2250 $\pm$ 1.0075\\
local & 8.1350 $\pm$ 0.8257 & 7.9900 $\pm$ 0.7647 & 8.1516 $\pm$ 0.8416 & 8.0152 $\pm$ 0.7813\\
\addlinespace
\multicolumn{5}{l}{\textbf{N=16, dense, $A$=5}}\\
random & 7.2500 $\pm$ 0.6230 & 6.8536 $\pm$ 0.4576 & 7.5083 $\pm$ 0.6361 & 7.0777 $\pm$ 0.5481\\
local & 8.0137 $\pm$ 0.5261 & 7.6063 $\pm$ 0.4808 & 8.3355 $\pm$ 0.5766 & 7.7962 $\pm$ 0.5070\\
\addlinespace
\multicolumn{5}{l}{\textbf{N=64, dense, $A$=2}}\\
random & 30.6068 $\pm$ 1.7884 & 30.2254 $\pm$ 1.8112 & 30.5768 $\pm$ 1.8269 & 30.3024 $\pm$ 1.8422\\
local & 34.4555 $\pm$ 2.3522 & 34.0041 $\pm$ 2.2973 & 34.4490 $\pm$ 2.3705 & 34.0307 $\pm$ 2.2954\\
\addlinespace
\multicolumn{5}{l}{\textbf{N=64, dense, $A$=5}}\\
random & 29.7438 $\pm$ 1.4425 & 29.5636 $\pm$ 1.5950 & 30.7895 $\pm$ 1.5686 & 30.0881 $\pm$ 1.7159\\
local & 34.3574 $\pm$ 1.3974 & 33.8925 $\pm$ 1.4584 & 35.9810 $\pm$ 1.6505 & 34.5690 $\pm$ 1.4825\\
\addlinespace
\multicolumn{5}{l}{\textbf{N=256, dense, $A$=2}}\\
random & 131.8977 $\pm$ 5.6838 & 131.5647 $\pm$ 5.6961 & 131.7268 $\pm$ 5.7418 & 131.6868 $\pm$ 5.7027\\
local & 146.7666 $\pm$ 8.2153 & 146.3067 $\pm$ 8.1567 & 146.8339 $\pm$ 8.2743 & 146.5392 $\pm$ 8.2072\\
\addlinespace
\multicolumn{5}{l}{\textbf{N=256, dense, $A$=5}}\\
random & 131.0364 $\pm$ 4.0811 & 132.5420 $\pm$ 4.4199 & 136.0641 $\pm$ 4.5217 & 134.7433 $\pm$ 4.6322\\
local & 147.7079 $\pm$ 8.4840 & 148.8172 $\pm$ 8.4700 & 154.9929 $\pm$ 9.5057 & 151.5419 $\pm$ 8.7438\\
\addlinespace
\multicolumn{5}{l}{\textbf{N=16, sparse, $A$=2}}\\
random & 6.2500 $\pm$ 0.5926 & 5.9695 $\pm$ 0.7438 & 6.2349 $\pm$ 0.5781 & 5.9979 $\pm$ 0.6978\\
local & 7.3652 $\pm$ 0.6958 & 7.0525 $\pm$ 0.6242 & 7.3459 $\pm$ 0.7019 & 7.0630 $\pm$ 0.5944\\
\addlinespace
\multicolumn{5}{l}{\textbf{N=16, sparse, $A$=5}}\\
random & 6.5777 $\pm$ 0.7425 & 6.1189 $\pm$ 0.7097 & 6.7667 $\pm$ 0.8417 & 6.2602 $\pm$ 0.7045\\
local & 7.6429 $\pm$ 0.6811 & 7.1963 $\pm$ 0.5757 & 7.9624 $\pm$ 0.7564 & 7.3153 $\pm$ 0.6223\\
\addlinespace
\multicolumn{5}{l}{\textbf{N=64, sparse, $A$=2}}\\
random & 39.7656 $\pm$ 1.5978 & 38.4418 $\pm$ 1.6634 & 39.8013 $\pm$ 1.6315 & 38.5579 $\pm$ 1.6117\\
local & 29.4037 $\pm$ 3.6337 & 28.0440 $\pm$ 3.2326 & 29.3561 $\pm$ 3.6566 & 28.0099 $\pm$ 3.3070\\
\addlinespace
\multicolumn{5}{l}{\textbf{N=64, sparse, $A$=5}}\\
random & 42.4158 $\pm$ 0.9151 & 41.4430 $\pm$ 1.0675 & 44.9540 $\pm$ 1.0624 & 42.8415 $\pm$ 1.3162\\
local & 29.1373 $\pm$ 2.8958 & 27.7146 $\pm$ 2.6391 & 29.9550 $\pm$ 3.0924 & 28.1664 $\pm$ 2.6459\\
\addlinespace
\multicolumn{5}{l}{\textbf{N=256, sparse, $A$=2}}\\
random & 163.5553 $\pm$ 2.9985 & 159.0313 $\pm$ 3.6277 & 163.7134 $\pm$ 3.0079 & 159.4412 $\pm$ 3.6833\\
local & 101.5544 $\pm$ 6.6979 & 97.5194 $\pm$ 6.4755 & 101.2302 $\pm$ 6.7571 & 97.5574 $\pm$ 6.4629\\
\addlinespace
\multicolumn{5}{l}{\textbf{N=256, sparse, $A$=5}}\\
random & 175.7398 $\pm$ 1.1899 & 172.1823 $\pm$ 2.3531 & 186.5906 $\pm$ 1.4152 & 177.2124 $\pm$ 2.3602\\
local & 107.2058 $\pm$ 6.3705 & 101.4155 $\pm$ 6.2523 & 109.8938 $\pm$ 7.0466 & 102.8164 $\pm$ 6.9903\\
\addlinespace
\end{longtable}

\end{document}